\documentclass[11pt]{scrartcl}

\usepackage{url}
\usepackage[hidelinks]{hyperref}
\usepackage[utf8]{inputenc}
\usepackage[small]{caption}
\usepackage{graphicx}
\usepackage{amsmath}
\usepackage{booktabs}
\usepackage{amsthm}
\usepackage{amssymb}

\usepackage{cleveref}
\usepackage{color}
\usepackage{xspace}
\usepackage{tikz}
\usepackage{tabularx}
\usepackage{pbox}
\usepackage{framed}
\usepackage{enumitem}
\usepackage{arydshln}

\DeclareOldFontCommand{\bf}{\normalfont\bfseries}{\mathbf}

\usepackage{pifont,xcolor}
\newcommand{\cmark}{\ding{51}}
\newcommand{\xmark}{\ding{55}}
\newcommand{\yes}{\textcolor{green!50!black}{\cmark}}
\newcommand{\no}{\textcolor{red!50!black}{\xmark}}

\newtheorem{theorem}{Theorem}[section]
\newtheorem{corollary}[theorem]{Corollary}
\newtheorem{proposition}[theorem]{Proposition}
\newtheorem{lemma}[theorem]{Lemma}

\theoremstyle{definition}
\newtheorem{definition}[theorem]{Definition}

\usepackage{natbib}

\allowdisplaybreaks

\begin{document}

\title{Picking Sequences and Monotonicity in Weighted Fair Division}

\author{
Mithun Chakraborty\\University of Michigan
\and
Ulrike Schmidt-Kraepelin\\Technische Universit\"{a}t Berlin
\and
Warut Suksompong\\National University of Singapore
}

\date{\vspace{-3ex}}

\maketitle

\begin{abstract}
We study the problem of fairly allocating indivisible items to agents with different entitlements, which captures, for example, the distribution of ministries among political parties in a coalition government.
Our focus is on picking sequences derived from common apportionment methods, including five traditional divisor methods and the quota method.
We paint a complete picture of these methods in relation to known envy-freeness and proportionality relaxations for indivisible items as well as monotonicity properties with respect to the resource, population, and weights.
In addition, we provide characterizations of picking sequences satisfying each of the fairness notions, and show that the well-studied maximum Nash welfare solution fails resource- and population-monotonicity even in the unweighted setting.
Our results serve as an argument in favor of using picking sequences in weighted fair division problems.
\end{abstract}

\section{Introduction}

After a national election, the parties forming a coalition government are faced with the task of dividing the ministries among themselves.
How can they perform this task in a fair manner, taking into account both their preferences on ministries and the votes that they received in the election?

The study of fairly allocating resources to interested agents (in this case, parties), commonly known as \emph{fair division}, has a long history dating back several decades \citep{BramsTa96,Moulin03}.
Among the most prominent fairness criteria are \emph{envy-freeness}---no agent prefers another agent's allocated bundle over her own---and \emph{proportionality}---if there are $n$ agents, then every agent receives at least $1/n$ of her value for the entire resource.
These criteria implicitly assume that all agents have the same entitlement to the resource, an assumption that is made in the vast majority of the fair division literature, yet utterly fails in our ministry example as well as when allotting supplies to districts, organizations, or university departments, which typically have different sizes.
Fortunately, both envy-freeness and proportionality allow for taking the entitlements, or \emph{weights}, into account in a natural way.
For instance, if agent~A's weight is twice that of agent~B, then A will be satisfied with respect to weighted envy-freeness as long as she derives at least twice as much value for her own bundle as for B's bundle.
While such weight-based extensions of classical fairness concepts are appropriate for scenarios with different entitlements, they sometimes cannot be satisfied when allocating indivisible items like ministries (e.g., when every party places all of its value on the same important ministry).
Consequently, recent work has proposed relaxations including \emph{weighted envy-freeness up to one item (WEF1)} \citep{ChakrabortyIgSu20} and \emph{weighted proportionality up to one item (WPROP1)} \citep{AzizMoSa20}, each of which can always be fulfilled.

An attractive class of procedures for allocating items is the class of \emph{picking sequences}: these procedures let agents take turns picking their favorite items according to a prespecified order.
Picking sequences are intuitive, can be implemented efficiently, and help preserve privacy since each agent only has to reveal the picks in her turns as opposed to her full preferences. In fact, several methods for apportioning seats in a parliament---a setting commonly known as \emph{apportionment}---can be formulated as picking sequences \citep{BramsKa04,OlearyGrEl05}.\footnote{Note that apportionment is a special case of our setting where all items are identical \citep{BalinskiYo01}. Some apportionment methods cannot be formulated as picking sequences---see Section~\ref{subsec:relatedwork} for more details.}
For example, Adams' method assigns each pick to an agent $i$ who minimizes $t_i/w_i$, where $t_i$ and $w_i$ denote the number of times that agent~$i$ has picked so far and her weight, respectively.
\citet{BramsKa04} proposed using picking sequences to allocate ministries, noting that such sequences have been used in Northern Ireland and Denmark, and \citet{ChakrabortyIgSu20} showed that the allocation produced by Adams' method always fulfills WEF1 but not necessarily WPROP1. 
It is therefore an important question which fairness criteria, if any, are satisfied by picking sequences based on other prevalent apportionment methods.

In addition to fairness, another desirable set of properties for allocation procedures is monotonicity in terms of the parameters of the setting.
In particular, \emph{resource-monotonicity} means that whenever an extra item is added, no agent receives a lower utility as a result.
Similarly, \emph{population-monotonicity} stipulates that introducing an additional agent should not increase the utility of any existing agent, and \emph{weight-monotonicity} implies that when the weight of an agent increases, her utility does not go down.\footnote{Resource-monotonicity is known as \emph{house-monotonicity} in the context of apportionment; a violation of it is referred to as the \emph{Alabama paradox} \citep{BalinskiYo01}. Likewise, violations of (variants of) population- and weight-monotonicity are called the \emph{new states paradox} and the \emph{population paradox}, respectively.}
\citet{SegalhaleviSz19} showed that for divisible items in the unweighted setting, the \emph{maximum Nash welfare (MNW)} solution, which chooses an allocation maximizing the product of the agents' utilities, is resource- and population-monotone.
How do picking sequences and (a weighted generalization of) MNW perform with respect to monotonicity properties in the weighted allocation of indivisible items?

\subsection{Our Results}

In this paper, we conduct a thorough investigation of picking sequences based on common apportionment methods, as well as the maximum (weighted) Nash welfare solution, in relation to fairness and monotonicity properties.
In addition to WEF1 and WPROP1, we consider \emph{weak weighted envy-freeness up to one item (WWEF1)}, a weakening of WEF1 proposed by \citet{ChakrabortyIgSu20}.
For brevity, we say that an allocation rule satisfies a fairness notion if the allocation that it produces for an arbitrary input instance satisfies that notion.

We begin in Section~\ref{sec:pickseq} by establishing fundamental results on our properties in the context of picking sequences.
In particular, we define three consistency properties with respect to the resource, population, and weights---for example, resource-consistency means that whenever an item is added, the new picking sequence should simply be the old one with an additional pick appended at the end.
We show that resource- and population-consistency imply the respective monotonicity properties for any number of agents, while weight-consistency implies weight-monotonicity only for two agents.
In addition, for each fairness notion, we characterize the picking sequences whose output always satisfies that notion.

With this groundwork laid, we proceed to determine the properties satisfied by different picking sequences.
First, in Section~\ref{sec:divisor}, we consider the picking sequences derived from divisor methods.
These methods assign each pick to an agent~$i$ who minimizes the ratio $f(t_i)/w_i$ capturing the proportion between the number of times the agent has picked so far and the agent's weight, where the function $f$ varies from method to method.
We establish that all divisor methods satisfy resource- and population-monotonicity for any number of agents as well as weight-monotonicity for two agents; however, we construct examples showing that the five traditional divisor methods due to Adams, Jefferson, Webster, Hill, and Dean all fail weight-monotonicity when there are three agents.
On the fairness front, we exhibit that a broad class of divisor methods called \emph{weighted power-mean divisor methods}, which include all five traditional methods, satisfy WWEF1---this illustrates a concrete manner in which divisor methods guarantee fairness in the weighted item allocation setting. 
We complement this result by proving that  Adams' method is the only divisor method satisfying the stronger notion of WEF1 while Jefferson's is the only one fulfilling WPROP1.\footnote{As we discuss in Section~\ref{subsec:relatedwork}, \citet{ChakrabortyIgSu20} showed that no method can satisfy both WEF1 and WPROP1.}

Next, in Section~\ref{sec:quota}, we address the picking sequence derived from another important apportionment method: the \emph{quota method}.
While not itself a divisor method, the quota method has a definition similar to that of Jefferson's method, which uses the function $f(t) = t+1$, but also imposes a ``quota'' to determine each agent's eligibility.
We show that the quota method exhibits similar monotonicity behavior as the divisor methods, with the notable exception that it fails population-monotonicity.
As for fairness, like Jefferson's method, the quota method satisfies WWEF1 and WPROP1.
In fact, these two rules are the first to have been shown to satisfy both WWEF1 and WPROP1, to the best of our knowledge.

Finally, in Section~\ref{sec:MNW}, we examine the \emph{maximum weighted Nash welfare (MWNW)} solution, which is a natural generalization of the well-studied MNW solution to the weighted setting.
\citet{ChakrabortyIgSu20} already proved that MWNW satisfies WWEF1 but not WEF1; we show that it fails WPROP1.
We then present examples demonstrating that even in the unweighted setting (where MWNW reduces to MNW), the rule fails both resource- and population-monotonicity.
This result stands in stark contrast to the aforementioned result of \citet{SegalhaleviSz19} that MNW is resource- and population-monotone in the context of divisible items, and is perhaps even more striking given that MNW is known to fulfill several desirable properties \citep{CaragiannisKuMo19,HalpernPrPs20,AmanatidisBiFi21}.
On the positive side, MWNW satisfies weight-monotonicity for any number of agents, and is the only rule to do so among the ones we consider in this paper.

Our results are summarized in Table~\ref{table:summary}.
Overall, we believe that they serve as an argument in favor of using picking sequences in division problems with unequal entitlements in view of both fairness and monotonicity considerations.

\begin{table*}[!ht]
\begin{center}
\begin{tabular}{ c|ccc|ccc } 
  \hline
  & \multicolumn{3}{ c| }{Monotonicity} & \multicolumn{3}{ c }{Fairness} \\
  \hline
  & Res.-mon. & Pop.-mon. & Weight-mon. & WEF1 & WWEF1 & WPROP1 \\
  \hline
 Adams & \yes & \yes & \no & \yes$^\dagger$ & \yes$^\dagger$ & \no$^\dagger$  \\
 Jefferson & \yes & \yes & \no & \no & \yes & \yes \\
 Webster & \yes & \yes & \no & \no & \yes & \no \\
 Hill & \yes & \yes & \no & \no & \yes & \no \\
 Dean & \yes & \yes & \no & \no & \yes & \no \\ \hline
 Quota & \yes & \no & \no & \no & \yes & \yes \\ \hline
 MWNW & \no & \no & \yes & \no$^\dagger$ & \yes$^\dagger$ & \no \\ \hline
\end{tabular}
\end{center}
\caption{Summary of our results. 
 \citet{ChakrabortyIgSu20} showed that Adams' method satisfies WEF1 and WWEF1 but not WPROP1, while MWNW satisfies WWEF1 but not WEF1; we indicate these results with dagger symbols. 
All other results are new to this paper.
All rules satisfy weight-monotonicity in the case of two agents.
MWNW fails resource- and population-monotonicity even in the unweighted setting.}
\label{table:summary}
\end{table*}

\subsection{Further Related Work}
\label{subsec:relatedwork}

The fair allocation of indivisible items has received substantial recent attention, notably among researchers in artificial intelligence, multiagent systems, and computational social choice---see the surveys by \citet{BouveretChMa16}, \citet{Markakis17}, \citet{Aziz20}, and \citet{Walsh20}.
A large majority of work assumes that all agents have equal entitlements, in which case the notions \emph{envy-freeness up to one item (EF1)} \citep{LiptonMaMo04,Budish11} and \emph{proportionality up to one item (PROP1)} \citep{ConitzerFrSh17,AzizCaIg19} are often considered.
Both WEF1 and WWEF1 reduce to EF1 in the unweighted setting, while WPROP1 reduces to PROP1.
Even though EF1 implies PROP1, \citet{ChakrabortyIgSu20} showed that no rule can simultaneously satisfy WEF1 and WPROP1.
\citet{AzizMoSa20} gave a protocol satisfying WPROP1 along with the economic efficiency notion of \emph{Pareto optimality}, while \citet{BabaioffNiTa21} considered competitive equilibrium for agents with different budgets representing their weights.
\citet{FarhadiGhHa19} proposed a weighted version of \emph{maximin share fairness} \citep{Budish11,KurokawaPrWa18}, and \citet{AzizChLi19} studied the analogous notion for chores (i.e., items that yield negative utilities).
\citet{BabaioffEzFe21} introduced the notion of \emph{AnyPrice share} and showed that a constant fraction of this share can be guaranteed for agents with arbitrary weights.
Unequal entitlements have also been considered in the context of allocating \emph{divisible} items, also known as \emph{cake cutting} \citep{Segalhalevi19,CrewNaSp20,CsehFl20}. 

Like fair division, apportionment methods have given rise to a long line of work that analyzes their advantages and disadvantages according to various desiderata \citep{BalinskiYo01,Pukelsheim14}.
As \citet{BalinskiYo01} noted, Adams' method tends to favor agents with smaller weights and Jefferson's typically benefits those with larger weights, whereas the other three divisor methods lie in between.
While all divisor methods and some non-divisor methods (e.g., the quota method) are associated with picking sequences, other non-divisor methods such as Hamilton's method do not give rise to a picking sequence and are therefore not useful in our context \citep[p.~149]{BramsKa04}.
Apportionment has also attracted interest in artificial intelligence \citep{BrillLaSk17,BrillGoPe20,BredereckFaFu20} as well as in philosophy \citep{WinteinHe18}.

Finally, picking sequences have been studied by several authors due to their simplicity and practicality \citep{BouveretLa11,BouveretLa14,AzizWaXi15,TominagaToYo16,BeynierBoLe19,XiaoLi20}, with a number of authors investigating manipulation issues.
We assume in this paper that agents are not strategic and always pick their most preferred item available.
In the unweighted setting, a popular picking sequence is the \emph{round-robin algorithm}, which lets agents pick items in cyclic order until the items run out; it is well-known that the output of the round-robin algorithm always satisfies EF1 under additive utilities.

\section{Preliminaries}
\label{sec:prelim}

We consider a discrete resource allocation setting with a set of agents $N=[n]$ and a set of indivisible items $M=[m]$, where $[k] := \{1,2,\dots,k\}$ for any $k\in\mathbb{N}$.
Each agent $i\in N$ is endowed with a \emph{weight} $w_i > 0$ and a \emph{utility function} $u_i: 2^M\rightarrow \mathbb{R}_{\ge 0}$; for convenience, we sometimes write $u_i(j)$ instead of $u_i(\{j\})$ for an item $j\in M$.
As is very common in the fair division literature, we assume that the utility functions are additive, i.e., $u_i(M')=\sum_{j\in M'}u_i(j)$ for all $i\in N$ and $M'\subseteq M$.
An \emph{allocation} $\mathcal{M}=(M_1,\dots,M_n)$ is a partition of the items into $n$ bundles so that agent $i$ receives bundle $M_i$.
An \emph{instance} consists of the agents, items, weights, and utility functions.
When all weights are equal (in which case we can take them to be~$1$ without loss of generality), we refer to the resulting setting as the \emph{unweighted setting}.

We consider the following three fairness notions. 
The first two notions were proposed by \citet{ChakrabortyIgSu20} and the third by \citet{AzizMoSa20}.
\begin{definition}
An allocation $(M_1,\dots,M_n)$ is said to satisfy
\begin{itemize}
\item \emph{weighted envy-freeness up to one item (WEF1)} if for any $i,j\in N$, there exists $B\subseteq M_j$ with $|B|\le 1$ such that 
$\frac{u_i(M_i)}{w_i} \geq \frac{u_i(M_j\setminus B)}{w_j}$;

\item \emph{weak weighted envy-freeness up to one item (WWEF1)} if for any $i,j\in N$, there exists $B\subseteq M_j$ with $|B|\le 1$ such that
$\frac{u_i(M_i)}{w_i} \geq \frac{u_i(M_j\setminus B)}{w_j}$
or
$
\frac{u_i(M_i\cup B)}{w_i} \geq \frac{u_i(M_j)}{w_j}
$;
\item \emph{weighted proportionality up to one item (WPROP1)} if for any $i\in N$, there exists $B\subseteq M\setminus M_i$ with $|B|\le 1$ such that
$
u_i(M_i) \ge \left(\frac{w_i}{\sum_{i'\in N}w_{i'}}\cdot u_i(M)\right) - u_i(B)
$.
\end{itemize}
\end{definition}
\citet{ChakrabortyIgSu20} showed that no rule can simultaneously satisfy WEF1 and WPROP1. 
In particular, consider an instance where $m=n$ and every agent has a nonzero utility for every item.
Any WEF1 allocation has to assign exactly one item to each agent.\footnote{Otherwise an agent with no item will (weighted-)envy an agent with at least two items by more than one item.}
On the other hand, if a certain agent has the same utility for all items and a sufficiently larger weight than every other agent, WPROP1 will require this agent to receive at least $m-1$ items.

A \emph{domain} refers to a set of instances.
A domain may include all instances with any number of agents and items, weights, and utility functions, or it may only include---for example---all instances with two agents, or all instances with equal weights (this corresponds to the unweighted setting).
An \emph{allocation rule} is a function that maps each instance in a given domain to an allocation; it is said to satisfy a fairness notion if the allocation that it produces always fulfills that notion.
We now define the three monotonicity properties that we consider---the first two have been studied by \citet{SegalhaleviSz18,SegalhaleviSz19}, while the third has not been studied in fair division to the best of our knowledge.
\begin{definition}\label{def:mon}
An allocation rule $\mathcal{R}$ with domain $\mathcal{I}$ satisfies
\begin{itemize}
\item \emph{resource-monotonicity} if the following holds: for any instance with $m$ items, when an extra item is added as item $m+1$, if both the original and the modified instance belong to $\mathcal{I}$, then each agent receives no higher utility from the allocation produced by $\mathcal{R}$ in the original instance than in the modified instance;
\item \emph{population-monotonicity} if the following holds: for any instance with $n$ agents, when an extra agent is added as agent $n+1$, if both the original and the modified instance belong to $\mathcal{I}$, then each of the first $n$ agents receives at least as much utility from the allocation produced by $\mathcal{R}$ in the original instance as in the modified instance;
\item \emph{weight-monotonicity} if the following holds: for any instance, when the weight of an agent increases, if both the original and the modified instance belong to $\mathcal{I}$, the utility that the agent receives from the allocation produced by $\mathcal{R}$ does not decrease.
\end{itemize}
\end{definition}
While these monotonicity properties are intuitive and it may seem that any reasonable allocation rule should satisfy them, this is in fact not the case: In Appendix~\ref{app:resmon-failure}, we show that two popular fair division algorithms---the \emph{envy cycle elimination algorithm} and the \emph{adjusted winner procedure}---fail resource-monotonicity even in the unweighted setting.

Next, we provide definitions related to picking sequences.

\begin{definition}
A \emph{picking sequence} on $n$ agents and $m$ items is a sequence $\pi_{n,m,\textbf{w}} = (a_1,a_2,\dots,a_m)$, where $a_i\in N$ for each $i\in M$.
A \emph{family of picking sequences} is a collection $\Pi = \{\pi_{n,m,\textbf{w}}\}$, with at most one picking sequence for each pair of positive integers $n,m$ and weight vector $\textbf{w}=(w_1,\dots,w_n)$.
A family of picking sequences $\Pi$ is called
\begin{itemize}
\item \emph{resource-consistent} if for every $n,m,\textbf{w}$ such that both $\pi_{n,m,\textbf{w}}$ and $\pi_{n,m+1,\textbf{w}}$ belong to $\Pi$, the sequence $\pi_{n,m,\textbf{w}}$ forms a prefix of $\pi_{n,m+1,\textbf{w}}$;
\item \emph{population-consistent} if for every $n,m,\textbf{w},\textbf{w}'$ such that $\textbf{w}' = (w_1,\dots,w_n,w_{n+1}')$ where $w_{n+1}'$ is the weight of agent $n+1$ and both $\pi_{n,m,\textbf{w}}$ and $\pi_{n+1,m,\textbf{w}'}$ belong to $\Pi$, the sequence $\pi_{n+1,m,\textbf{w}'}$ can be obtained from $\pi_{n,m,\textbf{w}}$ by inserting agent $n+1$ in some positions (possibly none) and trimming the suffix of the resulting sequence so that the sequence has length $m$.
\item \emph{weight-consistent} if the following holds: for every $n,m,\textbf{w}$ and $w_i' > w_i$ such that both $\pi_{n,m,\textbf{w}}$ and $\pi_{n,m,\textbf{w}'}$ belong to $\Pi$, where $\textbf{w}' = (w_1,\dots,w_i',\dots,w_n)$, the sequence $\pi_{n,m,\textbf{w}'}$ can be obtained from $\pi_{n,m,\textbf{w}}$ by moving some of agent $i$'s picks earlier (possibly none), inserting agent $i$ in some positions (possibly none), and trimming the suffix of the resulting sequence so that the sequence has length $m$.
\end{itemize}
\end{definition}

Given a picking sequence $(a_1,\dots,a_m)$ and the agents' utility functions, we assume that in the $i$th turn, agent $a_i$ picks her highest-valued item from among the remaining items, breaking ties in a consistent manner (say, in favor of lower-numbered items).
We sometimes drop the subscript from $\pi_{n,m,\textbf{w}}$ when $n,m,\textbf{w}$ are clear from the context.
A family of picking sequences generates an allocation rule, which we will refer to interchangeably with the family itself.
We also refer to a picking sequence $\pi$ interchangeably with the family of picking sequences that consists only of $\pi$.

\section{General Picking Sequences}
\label{sec:pickseq}

We begin by proving results for general picking sequences.
In addition to being interesting in their own right, these results will later help us determine the properties that each apportionment method satisfies.
First, we present characterizations of picking sequences whose output is guaranteed to satisfy each of the fairness notions WEF1, WWEF1, and WPROP1.
In particular, we show that a picking sequence guarantees a fairness notion for agents with arbitrary utility functions if and only if it does so for agents with identical utility functions that put utility $1$ on some items and $0$ on the remaining items.
This means that the fairness guarantees for general additive utilities can be expressed as relatively simple conditions on the number of picks in each prefix of the picking sequence.

\begin{theorem}
\label{thm:WEF1}
A picking sequence $\pi$ satisfies WEF1 if and only if for every prefix of $\pi$ and every pair of agents $i,j$ with $t_j\ge 2$, we have $\frac{t_i}{t_j-1}\ge \frac{w_i}{w_j}$, where $t_i$ and $t_j$ denote the number of agent $i$'s and agent $j$'s picks in the prefix, respectively.
\end{theorem}

\begin{proof}
($\Leftarrow$) The backward direction was already established by \citet[Lemma~3.5]{ChakrabortyIgSu20}.

($\Rightarrow$) For the forward direction, we prove the contrapositive.
Assume that the inequality does not hold for some prefix of length $k$ and some agents $i,j \in N$.
Consider an instance where every agent has utility $1$ for each of the first $k$ items and $0$ for the remaining $m-k$ items.
The $k$ valued items are picked in the first $k$ turns, so agent $i$ has utility $t_i$ for bundle $M_i$.
On the other hand, after removing any item from $M_j$, agent $i$ still has value $t_j-1$ for the remaining bundle.
Hence, for any $B\subseteq M_j$ with $|B|\le 1$, we have
$\frac{u_i(M_i)}{u_i(M_j\setminus B)}
\leq \frac{t_i}{t_j-1} < \frac{w_i}{w_j}$,
implying that $\pi$ is not WEF1.
\end{proof}

\begin{theorem}
\label{thm:WWEF1}
A picking sequence $\pi$ satisfies WWEF1 if and only if for every prefix of $\pi$ and every pair of agents $i,j$ with $t_j\ge 2$, both of the following conditions hold:
\begin{itemize}
\item $\frac{t_i}{t_j-1}\ge \frac{w_i}{w_j}$ if $w_i\ge w_j$;
\item $\frac{t_i+1}{t_j}\ge \frac{w_i}{w_j}$ if $w_i\le w_j$,
\end{itemize}
where $t_i$ and $t_j$ denote the number of agent $i$'s and agent $j$'s picks in the prefix, respectively.
\end{theorem}

\begin{proof}
($\Rightarrow$) We begin with the forward direction.
Assume for contradiction that $\pi$ fulfills WWEF1 but fails (at least) one of the conditions for some prefix of length $k$ and some agents $i,j \in N$ with $t_j \geq 2$.
Consider an instance where every agent has utility $1$ for each of the first $k$ items and $0$ for the remaining $m-k$ items.
The $k$ valued items are picked in the first $k$ turns.

If the first condition fails for the prefix of length $k$, then for any $B\subseteq M_j$ with $|B|\le 1$, we have 
$\frac{u_i(M_i)}{u_i(M_j\setminus B)}
\leq \frac{t_i}{t_j-1} < \frac{w_i}{w_j}$ and
$\frac{u_i(M_i\cup B)}{u_i(M_j)}\leq \frac{t_i+1}{t_j}< \frac{w_i}{w_j}$.
To see the last inequality, note that if $t_i+1\geq t_j$, then $\frac{t_i+1}{t_j} \leq \frac{t_i}{t_j-1} < \frac{w_i}{w_j}$; else, $\frac{t_i+1}{t_j} < 1\leq\frac{w_i}{w_j}$.
So $\pi$ is not WWEF1, a contradiction.

If the second condition fails for the prefix of length $k$, then since $\frac{t_i+1}{t_j} < \frac{w_i}{w_j}\leq 1$, we have $\frac{t_i}{t_j-1} < \frac{t_i+1}{t_j} < \frac{w_i}{w_j}$.
Hence, for any $B\subseteq M_j$ with $|B|\le 1$, the same two chains of inequalities in the previous paragraph hold, and again $\pi$ is not WWEF1, a contradiction.

($\Leftarrow$) We now proceed to the backward direction.
In the case where $w_i\ge w_j$, \citet[Lemma~3.5]{ChakrabortyIgSu20} already showed that the WEF1 condition (and therefore the WWEF1 condition) is fulfilled for agent $i$ towards agent~$j$.
Assume therefore that $w_i < w_j$.
It suffices to show that the WWEF1 condition for agent $i$ towards agent $j$ is fulfilled after every pick by agent $j$.

Let $\gamma = w_i/w_j$, and consider any pick by agent $j$---suppose it is the agent's $t_j$-th pick, where $t_j\ge 2$.
Denote by $\tau_1$ the number of times agent $i$ picks before agent $j$'s first pick, and $\tau_\ell$ the number of times agent $i$ picks between agent $j$'s $(\ell-1)$st and $\ell$-th picks for $2\le \ell\le t_j$.
Let agent $i$'s utility for the items  that she picks herself in phase $\ell$ be $\alpha_1^\ell, \alpha_2^\ell, \dots,\alpha_{\tau_\ell}^\ell$, respectively, where the phases $1,2,\dots,t_j$ are defined as in the previous sentence.
Moreover, let agent $i$'s utility for the items that agent $j$ picks be $\beta_1,\dots,\beta_{t_j}$.
The condition in the theorem statement implies that
\begin{equation*}
1+\sum_{\ell=1}^s \tau_\ell \ge s\gamma
\qquad\qquad \forall s\in\{2,\dots,t_j\}.
\end{equation*}
Note that for $t_j = 1$, we trivially have $\frac{t_i+1}{t_j}\ge 1\ge\frac{w_i}{w_j}$.
Hence, 
\begin{equation}
\label{eq:WWEF-picks}
1+\sum_{\ell=1}^s \tau_\ell \ge s\gamma
\qquad\qquad \forall s\in[t_j].
\end{equation}
Every time agent $i$ picks, she picks an item with the highest utility for her among the available items, which means that
\begin{equation}
\label{eq:WWEF-utils}
\sum_{x=1}^{\tau_\ell}\alpha_x^{\ell} \geq \tau_\ell\max\{\beta_\ell, \beta_{\ell+1},\dots,\beta_{t_j}\}
\qquad \forall \ell\in[t_j].
\end{equation}
Note that (\ref{eq:WWEF-utils}) holds trivially if $\tau_\ell=0$ since both sides are zero.

We claim that for each $1\leq s\leq t_j$, 
\begin{align*}
&\max\{\beta_1,\dots,\beta_{t_j}\} + \sum_{\ell=1}^s\sum_{x=1}^{\tau_{\ell}}\alpha_x^\ell \ge \gamma\sum_{\ell=1}^s\beta_\ell + \left(1+\sum_{\ell=1}^s\tau_\ell - s\gamma\right)\max\{\beta_s,\beta_{s+1},\dots,\beta_{t_j}\}.
\end{align*}
To prove the claim, we use induction on $s$.
For the base case $s=1$, we have from (\ref{eq:WWEF-utils}) that
\begin{align*}
\max\{\beta_1,\dots,\beta_{t_j}\} + \sum_{x=1}^{\tau_1}\alpha_x^1 &\ge (1+\tau_1) \max\{\beta_1,\dots,\beta_{t_j}\} \\
&\ge \gamma\beta_1 + (1+\tau_1-\gamma)\max\{\beta_1,\dots,\beta_{t_j}\}.
\end{align*}
For the inductive step, assume that the claim holds for some $s-1\ge 1$; we will prove it for $s$.
We have
\begin{align*}
\max&\{\beta_1,\dots,\beta_{t_j}\} + \sum_{\ell=1}^s\sum_{x=1}^{\tau_{\ell}}\alpha_x^\ell \\
&= \max\{\beta_1,\dots,\beta_{t_j}\} + \sum_{\ell=1}^{s-1}\sum_{x=1}^{\tau_{\ell}}\alpha_x^\ell + \sum_{x=1}^{\tau_s}\alpha_x^s \\
&\ge \gamma\sum_{\ell=1}^{s-1}\beta_\ell + \left(1+\sum_{\ell=1}^{s-1}\tau_\ell - (s-1)\gamma\right)\max\{\beta_{s-1},\beta_s,\dots,\beta_{t_j}\} + \sum_{x=1}^{\tau_s}\alpha_x^s \\
&\ge \gamma\sum_{\ell=1}^{s-1}\beta_\ell + \left(1+\sum_{\ell=1}^{s-1}\tau_\ell - (s-1)\gamma\right)\max\{\beta_{s-1},\beta_s,\dots,\beta_{t_j}\} \\
&\qquad+ \tau_s\max\{\beta_s,\beta_{s+1},\dots,\beta_{t_j}\} \\
&\ge \gamma\sum_{\ell=1}^{s-1}\beta_\ell + \left(1+\sum_{\ell=1}^{s-1}\tau_\ell - (s-1)\gamma\right)\max\{\beta_s,\beta_{s+1},\dots,\beta_{t_j}\} \\
&\qquad+ \tau_s\max\{\beta_s,\beta_{s+1},\dots,\beta_{t_j}\} \\
&= \gamma\sum_{\ell=1}^{s-1}\beta_\ell + \left(1+\sum_{\ell=1}^s\tau_\ell - (s-1)\gamma\right)\max\{\beta_s,\beta_{s+1},\dots,\beta_{t_j}\} \\
&= \gamma\sum_{\ell=1}^{s-1}\beta_\ell + \gamma\max\{\beta_s,\beta_{s+1},\dots,\beta_{t_j}\}  + \left(1+\sum_{\ell=1}^s\tau_\ell - s\gamma\right)\max\{\beta_s,\beta_{s+1},\dots,\beta_{t_j}\} \\
&\ge \gamma\sum_{\ell=1}^{s-1}\beta_\ell + \gamma\beta_s  + \left(1+\sum_{\ell=1}^s\tau_\ell - s\gamma\right)\max\{\beta_s,\beta_{s+1},\dots,\beta_{t_j}\} \\
&= \gamma\sum_{\ell=1}^s\beta_\ell + \left(1+\sum_{\ell=1}^s\tau_\ell - s\gamma\right)\max\{\beta_s,\beta_{s+1},\dots,\beta_{t_j}\}.
\end{align*}
Here, the first inequality follows from the inductive hypothesis, the second from (\ref{eq:WWEF-utils}), and the third from (\ref{eq:WWEF-picks}) and the fact that $\max\{\beta_{s-1},\beta_s,\dots,\beta_{t_j}\} \ge \max\{\beta_s,\beta_{s+1},\dots,\beta_{t_j}\}$.
This completes the induction and establishes the claim.

Finally, taking $s=t_j$ in the claim, we get
\begin{align*}
\max\{\beta_1,\dots,\beta_{t_j}\} + \sum_{\ell=1}^{t_j}\sum_{x=1}^{\tau_{\ell}}\alpha_x^\ell &\ge \gamma\sum_{\ell=1}^{t_j}\beta_\ell + \left(1+\sum_{\ell=1}^{t_j}\tau_\ell - t_j\gamma\right)\beta_{t_j} \ge \gamma\sum_{\ell=1}^{t_j}\beta_\ell,
\end{align*}
where the second inequality follows from (\ref{eq:WWEF-picks}).
Letting $M_i$ and $M_j$ be the bundles of agents $i$ and $j$ after agent $j$'s $t_j$-th pick, and $B$ be a singleton set consisting of an item in $M_j$ for which agent $i$ has the highest utility, we have $u_i(B) + u_i(M_i) \ge \frac{w_i}{w_j}\cdot u_i(M_j)$.
Hence the WWEF1 condition for agent $i$ towards agent $j$ is fulfilled, as desired.
\end{proof}

As an example of a picking sequence that satisfies WWEF1 but not WEF1, suppose that $n=2$, $w_1=1$, $w_2=2$, and consider the sequence $(1,2,2,2,2)$.
For this sequence, we have $t_1=1$ and $t_2=4$, and therefore $\frac{t_1+1}{t_2} = \frac{w_1}{w_2} > \frac{t_1}{t_2-1}$.

\begin{theorem}
\label{thm:WPROP1}
A picking sequence $\pi$ satisfies WPROP1 if and only if for every prefix of $\pi$ and every agent $i$, we have $t_i\ge \left(\frac{w_i}{\sum_{i'\in N}w_{i'}}\cdot k\right) - 1$, where $t_i$ and $k$ denote the number of agent $i$'s picks in the prefix and the length of the prefix, respectively.
\end{theorem}

\begin{proof}
($\Rightarrow$) We begin with the forward direction.
Assume for contradiction that $\pi$ fulfills WPROP1 but the inequality does not hold for some prefix of length $k$ and some agent $i \in N$.
Consider an instance where every agent has utility $1$ for each of the first $k$ items and $0$ for the remaining $m-k$ items.
The $k$ valued items are picked in the first $k$ turns, so agent $i$ has utility $t_i$ for bundle $M_i$ and $k$ for the entire set of items.
Hence, for any $B\subseteq M\setminus M_i$ with $|B|\le 1$, we have
\begin{align*}
u_i(M_i) = t_i  
&< \left(\frac{w_i}{\sum_{i'\in N}w_{i'}}\cdot k\right) - 1 \le \left(\frac{w_i}{\sum_{i'\in N}w_{i'}}\cdot u_i(M)\right) - u_i(B),
\end{align*}
implying that $\pi$ is not WPROP1, a contradiction.

($\Leftarrow$) We now proceed to the backward direction.
Let $\gamma = \frac{w_i}{\sum_{i'\in N}w_{i'}}$, and assume that in $\pi$, agent $i$ picks $t_i$ times.
Denote by $\tau_0$ the number of times other agents pick before agent~$i$'s first pick, and $\tau_\ell$ the number of times other agents pick after agent $i$'s $\ell$-th pick and before her $(\ell+1)$st pick for $1\leq\ell\leq t_i$, where we insert a dummy empty pick for agent $i$ at the end of the picking sequence.
Let agent $i$'s utility for the items that other agents pick in phase $\ell$ be $\beta^\ell_1,\beta^\ell_2,\dots,\beta^\ell_{\tau_\ell}$, where the phases $0,1,\dots,t_i$ are defined as in the previous sentence.
Moreover, let agent $i$'s utility for the items that she picks be $\alpha_1,\dots,\alpha_{t_i}$, respectively.
We may assume that at least one other agent has at least one pick; otherwise the WPROP1 condition holds trivially.
Let $B$ be a singleton set consisting of an item in $M\setminus M_i$ for which agent $i$ has the highest utility.
For the prefix before agent $i$'s first pick, the condition in the theorem statement implies that $0\ge \gamma\tau_0-1$, i.e., $\tau_0 \le 1/\gamma$.
By definition of $B$, we have $u_i(B)\ge \beta^\ell_j$ for all $0\leq\ell\leq t_i$ and $1\leq j\leq \tau_\ell$, so $\tau_0u_i(B)\ge \sum_{j=1}^{\tau_0}\beta^0_j$.
It follows that 
\begin{align}
\label{eq:WPROP-init}
u_i(B) &= (\tau_0\gamma + (1-\tau_0\gamma))\cdot u_i(B) \ge \gamma\sum_{j=1}^{\tau_0}\beta^0_j + (1-\tau_0\gamma)\max_{\substack{0\le \ell\le t_i \\ 1\le j\le \tau_\ell}}\beta^\ell_j.
\end{align}

We claim that for each $0\leq s\leq t_i$,
\begin{align*}
&u_i(B) + (1-\gamma)\sum_{\ell=1}^s \alpha_\ell \ge \gamma\sum_{\ell=0}^s\sum_{j=1}^{\tau_\ell}\beta^\ell_j + \left(1+s(1-\gamma)-\gamma\sum_{\ell=0}^s \tau_\ell\right)\max_{\substack{s\le \ell\le t_i \\ 1\le j\le \tau_\ell}}\beta^\ell_j.
\end{align*}
To prove the claim, we use induction on $s$.
The base case $s=0$ follows immediately from (\ref{eq:WPROP-init}) since the terms $(1-\gamma)\sum_{\ell=1}^s \alpha_\ell$ and $s(1-\gamma)$ vanish.
For the inductive step, assume that the claim holds for some $s-1\ge 0$; we will prove it for $s$.
For the prefix before agent $i$'s $s$-th pick, the condition in the theorem statement implies that $s-1 \ge \gamma\left(s-1 + \sum_{\ell=0}^{s-1}\tau_\ell\right)-1$, or equivalently
\begin{equation}
\label{eq:WPROP-picks}
1+(s-1)(1-\gamma)\ge \gamma\sum_{\ell=0}^{s-1} \tau_\ell.
\end{equation}
Every time agent $i$ picks, she picks an item with the highest utility for her among the available items, which means that
\begin{equation}
\label{eq:WPROP-utils}
\alpha_s \geq \beta^s_j
\qquad\qquad \forall j\in[\tau_s].
\end{equation}
We have
\begin{align*}
u_i(B) &+ (1-\gamma)\sum_{\ell=1}^s \alpha_\ell \\
&= u_i(B) + (1-\gamma)\sum_{\ell=1}^{s-1}\alpha_\ell + (1-\gamma)\alpha_s \\
&\ge \gamma\sum_{\ell=0}^{s-1}\sum_{j=1}^{\tau_\ell}\beta^\ell_j + \left(1+(s-1)(1-\gamma)-\gamma\sum_{\ell=0}^{s-1} \tau_\ell\right) \max_{\substack{s-1\le \ell\le t_i \\ 1\le j\le \tau_\ell}} \beta^\ell_j + (1-\gamma)\alpha_s \\
&\ge \gamma\sum_{\ell=0}^{s-1}\sum_{j=1}^{\tau_\ell}\beta^\ell_j + \left(1+(s-1)(1-\gamma)-\gamma\sum_{\ell=0}^{s-1} \tau_\ell\right)\max_{\substack{s\le \ell\le t_i \\ 1\le j\le \tau_\ell}}\beta^\ell_j + (1-\gamma)\alpha_s \\
&= \gamma\sum_{\ell=0}^{s-1}\sum_{j=1}^{\tau_\ell}\beta^\ell_j +  \left(1+s(1-\gamma)-\gamma\sum_{\ell=0}^s \tau_\ell\right)\max_{\substack{s\le \ell\le t_i \\ 1\le j\le \tau_\ell}}\beta^\ell_j \\
&\qquad + \left(\gamma\tau_s - (1-\gamma)\right)\max_{\substack{s\le \ell\le t_i \\ 1\le j\le \tau_\ell}}\beta^\ell_j + (1-\gamma)\alpha_s \\
&\ge \gamma\sum_{\ell=0}^{s-1}\sum_{j=1}^{\tau_\ell}\beta^\ell_j +  \left(1+s(1-\gamma)-\gamma\sum_{\ell=0}^s \tau_\ell\right)\max_{\substack{s\le \ell\le t_i \\ 1\le j\le \tau_\ell}}\beta^\ell_j \\
&\qquad + \frac{\gamma\tau_s - (1-\gamma)}{\tau_s}\cdot\sum_{j=1}^{\tau_s}\beta_j^s + \frac{1-\gamma}{\tau_s}\cdot\sum_{j=1}^{\tau_s}\beta_j^s \\
&= \gamma\sum_{\ell=0}^{s-1}\sum_{j=1}^{\tau_\ell}\beta^\ell_j +  \left(1+s(1-\gamma)-\gamma\sum_{\ell=0}^s \tau_\ell\right)\max_{\substack{s\le \ell\le t_i \\ 1\le j\le \tau_\ell}}\beta^\ell_j  + \gamma\sum_{j=1}^{\tau_s}\beta_j^s \\
&= \gamma\sum_{\ell=0}^{s}\sum_{j=1}^{\tau_\ell}\beta^\ell_j + \left(1+s(1-\gamma)-\gamma\sum_{\ell=0}^s \tau_\ell\right)\max_{\substack{s\le \ell\le t_i \\ 1\le j\le \tau_\ell}}\beta^\ell_j.
\end{align*}
Here, the first inequality follows from the inductive hypothesis, the second from (\ref{eq:WPROP-picks}) and the fact that $\max_{\substack{s-1\le \ell\le t_i \\ 1\le j\le \tau_\ell}}\beta^\ell_j \ge \max_{\substack{s\le \ell\le t_i \\ 1\le j\le \tau_\ell}}\beta^\ell_j$, and the third from (\ref{eq:WPROP-utils}).
This completes the induction and establishes the claim.

Finally, taking $s=t_i$ in the claim, we get
\begin{align*}
u_i(B) + (1-\gamma)\sum_{\ell=1}^{t_i} \alpha_\ell &\ge \gamma\sum_{\ell=0}^{t_i}\sum_{j=1}^{\tau_\ell}\beta^\ell_j + \left(1+t_i(1-\gamma)-\gamma\sum_{\ell=0}^{t_i} \tau_\ell\right)\max_{1\le j\le \tau_{t_i}}\beta^\ell_j \\
&\ge \gamma\sum_{\ell=0}^{t_i}\sum_{j=1}^{\tau_\ell}\beta^\ell_j,
\end{align*}
where the second inequality follows from the condition in the theorem statement for the entire picking sequence.
Adding $\gamma\sum_{\ell=1}^{t_i} \alpha_\ell$ to both sides, we find that
\[
u_i(B) + \sum_{\ell=1}^{t_i} \alpha_\ell \ge \gamma\left(\sum_{\ell=1}^{t_i} \alpha_\ell + \sum_{\ell=0}^{t_i}\sum_{j=1}^{\tau_\ell}\beta^\ell_j\right).
\]
In other words, $u_i(B) + u_i(M_i) \ge \frac{w_i}{\sum_{i'\in N}w_{i'}}\cdot u_i(M)$.
This means that the allocation produced by $\pi$ is WPROP1, as desired.
\end{proof}

Next, we establish a strong relationship between resource- and population-consistency and the corresponding monotonicity notions.

\begin{theorem}
\label{thm:resmon}
Any resource-consistent family of picking sequences satisfies resource-monotonicity.
\end{theorem}

\begin{proof}
Let $\Pi$ be a resource-consistent family of picking sequences, and consider any $n,m,\textbf{w}$ such that both $\pi_{n,m,\textbf{w}}$ and $\pi_{n,m+1,\textbf{w}}$ belong to $\Pi$; since we will not vary the agents or weights, we drop $n$ and $\textbf{w}$ from the subscripts.
By resource-consistency, $\pi_{m}=(a_1,\dots,a_m)$ forms a prefix of $\pi_{m+1}=(a_1,\dots,a_m,a_{m+1})$.
For each $j\in[m]$, denote by $M_{\text{old},j}$ and $M_{\text{new},j}$ the set of available items at agent $a_j$'s pick in $\pi_{m}$ and $\pi_{m+1}$, respectively.

We prove by induction that $M_{\text{old},j}\subseteq M_{\text{new},j}$ for every $j\in[m]$.
The base case $j=1$ holds since we add an item and do not remove any when going from the instance for $\pi_{m}$ to that for $\pi_{m+1}$.
Assume that the subset relation holds for some $j\in[m-1]$.
If agent $a_j$ picks an item from $M_{\text{new},j}\setminus M_{\text{old},j}$ in $\pi_{m+1}$, then we immediately have $M_{\text{old},j+1}\subseteq M_{\text{old},j}\subseteq M_{\text{new},j+1}$.
Else, agent~$a_j$ picks an item $\ell$ from $M_{\text{old},j}$ in $\pi_{m+1}$.
In this case, since $M_{\text{old},j}\subseteq M_{\text{new},j}$ and ties are broken in a consistent manner, agent $a_j$ must also pick the same item $\ell$ in $\pi_{m}$.
We therefore have $M_{\text{old},j+1} = M_{\text{old},j}\setminus\{\ell\}\subseteq M_{\text{new},j}\setminus\{\ell\} = M_{\text{new},j+1}$, completing the induction.

Now, $M_{\text{old},j}\subseteq M_{\text{new},j}$ implies that agent $a_j$ receives at least as much utility from this pick in $\pi_{m+1}$ as from the corresponding pick in $\pi_{m}$.
Since $\pi_{m}$ forms a prefix of $\pi_{m+1}$ and utilities are additive, it follows that every agent is no worse off in the instance for $\pi_{m+1}$ than for $\pi_{m}$.
This means that $\Pi$ is resource-monotone, as desired.
\end{proof}

\begin{theorem}
\label{thm:popmon}
Any population-consistent family of picking sequences satisfies population-monotonicity.
\end{theorem}

\begin{proof}
Let $\Pi$ be a population-consistent family of picking sequences, and consider any $n,m,\textbf{w}$ such that both $\pi_{n,m,\textbf{w}}$ and $\pi_{n+1,m,\textbf{w}}$ belong to $\Pi$; since we will not vary the items or weights, we drop $m$ and $\textbf{w}$ from the subscripts.
Let $\pi_{n} = (a_1,\dots,a_m)$.
By population-consistency, $\pi_{n+1}$ can be obtained from $\pi_{n}$ by inserting agent~$n+1$ in some positions and trimming the suffix of the resulting sequence.
For each $j\in[m]$, denote by $M_{\text{old},j}$ the set of available items at agent $a_j$'s pick in $\pi_{n}$, and by $M_{\text{new},j}$ the set of available items at agent $a_j$'s corresponding pick in $\pi_{n+1}$ after we have inserted agent $n+1$ in some positions of $\pi_{n}$.
(If the corresponding pick does not exist in $\pi_{n+1}$ because it has been trimmed, we simply let $M_{\text{new},j}=\emptyset$.)

We prove by induction that $M_{\text{old},j}\supseteq M_{\text{new},j}$ for every $j\in[m]$.
The base case $j=1$ holds since $M_{\text{old},j}$ consists of all $m$ items.
Assume that the superset relation holds for some $j\in[m-1]$.
If agent $a_j$ picks an item from $M_{\text{old},j}\setminus M_{\text{new},j}$ in $\pi_{n}$, then we immediately have $M_{\text{old},j+1}\supseteq M_{\text{new},j}\supseteq M_{\text{new},j+1}$.
Else, agent $a_j$ picks an item $\ell$ from $M_{\text{new},j}$ in $\pi_{n}$.
In this case, since $M_{\text{old},j}\supseteq M_{\text{new},j}$ and ties are broken in a consistent manner, agent $a_j$ must also pick the same item $\ell$ in the corresponding pick of $\pi_{n+1}$.
We therefore have $M_{\text{old},j+1} = M_{\text{old},j}\setminus\{\ell\}\supseteq M_{\text{new},j}\setminus\{\ell\}\supseteq M_{\text{new},j+1}$, where the last relation is a superset instead of an equality because agent $n+1$ may get to pick between agent $a_j$'s pick and agent $a_{j+1}$'s pick in $\pi_{n+1}$.
This completes the induction.

Now, $M_{\text{old},j}\supseteq M_{\text{new},j}$ implies that agent $a_j$ receives at least as much utility from this pick in $\pi_{n}$ as from the corresponding pick in $\pi_{n+1}$ (if the latter pick does not exist in $\pi_{n+1}$, the utility is $0$).
Since utilities are additive, it follows that every agent is no better off in the instance for $\pi_{n+1}$ than for $\pi_{n}$.
This means that $\Pi$ is population-monotone, as desired. 
\end{proof}

The relationship between weight-consistency and weight-monotonicity is less straightforward: we show that the former implies the latter in the case of two agents.
As we will see later (Proposition~\ref{prop:weightmon-negative}), this relationship breaks down when there are three agents.

\begin{theorem}
\label{thm:weightmon-2}
For two agents, any weight-consistent family of picking sequences satisfies weight-monotonicity.
\end{theorem}

To establish the theorem, we first need the following lemma concerning a single adjacent swap in a picking sequence.

\begin{lemma}
\label{lem:weightmon-switch}
Suppose that there are two agents and $m$ items.
Let $\pi$ and $\pi'$ be picking sequences such that the only difference between them is that for some $1\leq i\leq m-1$, agent~$2$ picks in the $i$-th turn of $\pi$ and the $(i+1)$st turn of $\pi'$, whereas agent~$1$ picks in the $i$-th turn of $\pi'$ and the $(i+1)$st turn of $\pi$.
Then, for any utility functions of the two agents, agent~$1$ receives at least as much utility from $\pi'$ as from $\pi$.
\end{lemma}

\begin{proof}
Since $\pi$ and $\pi'$ are identical up to the $(i-1)$st turn and ties are broken in a consistent manner, picking proceeds identically for the two sequences, so we may ignore these picks.
Assume that in $\pi$, agent~$2$ picks item $\ell_2$ in the $i$-th turn and agent~$1$ picks item $\ell_1$ in the $(i+1)$st turn.
Observe that in the $i$-th turn of $\pi'$, agent~$1$ must pick either $\ell_1$ or $\ell_2$.
If agent~$1$ picks $\ell_1$, agent~$2$ will pick $\ell_2$ in the $(i+1)$st turn, and the remaining picks will be identical for both sequences.
Hence, in this case, agent~$1$'s utilities from $\pi$ and $\pi'$ are identical.
Suppose therefore that agent~$1$ picks $\ell_2$ in the $i$-th turn of $\pi'$, which means that $u_1(\ell_2)\ge u_1(\ell_1)$.
If agent~$2$ picks $\ell_1$ in the $(i+1)$st turn of $\pi'$, then the remaining picks will again be identical for both sequences, and agent~$1$'s utility from $\pi'$ is at least as high as from $\pi$.

Assume now that in $\pi'$, agent~$1$ picks $\ell_2$ in the $i$-th turn and agent~$2$ picks $\ell_3\ne\ell_1$ in the $(i+1)$st turn.
We claim that before each subsequent turn $j\ge i+2$ of the two sequences, if we let $M_{\pi,j}$ and $M_{\pi',j}$ be the respective set of remaining items, then exactly one of the following two conditions holds:
\begin{enumerate}[label=(\roman*)]
\item $M_{\pi,j} = M_{\pi',j}$;
\item $|M_{\pi,j}\setminus M_{\pi',j}| = |M_{\pi',j}\setminus M_{\pi,j}| = 1$, the item in $M_{\pi,j}\setminus M_{\pi',j}$ is agent~$2$'s highest-valued item in $M_{\pi,j}\cup M_{\pi',j}$ (taking tie-breaking into account), and the item in $M_{\pi',j}\setminus M_{\pi,j}$ is agent~$1$'s highest-valued item in $M_{\pi,j}\cup M_{\pi',j}$ (taking tie-breaking into account).
\end{enumerate}
We prove this claim by induction on $j$. For $j=i+2$, we have $M_{\pi,j}\setminus M_{\pi',j} = \{\ell_3\}$, and this item is agent~$2$'s highest-valued item in $M_{\pi,j}\cup M_{\pi',j}$ due to the agent's choice in the $(i+1)$st turn of $\pi'$; an analogous argument holds for $M_{\pi',j}\setminus M_{\pi,j}$.
Suppose that the claim holds for some $j\ge i+2$; we will prove it for $j+1$.
If condition~(i) holds for $j$, then it obviously also holds for $j+1$.
Assume that (ii) holds for $j$ and, without loss of generality, that agent~$1$ picks in the $j$-th turn.
Let $M_{\pi,j}\setminus M_{\pi',j} = \{\ell\}$ and $M_{\pi',j}\setminus M_{\pi,j} = \{\ell'\}$.
By condition~(ii), agent~$1$ picks item~$\ell'$ in the $j$-th turn of $\pi'$.
If agent~$1$ picks $\ell$ in the $j$-th turn of $\pi$, then $M_{\pi,j+1} = M_{\pi',j+1}$ and condition~(i) holds for $j+1$.
Hence, suppose that agent~$1$ picks $\ell''\not\in\{\ell,\ell'\}$ in the $j$-th turn of $\pi$.
In this case, we have $M_{\pi',j+1}\setminus M_{\pi,j+1} = \{\ell''\}$ and $M_{\pi,j+1}\setminus M_{\pi',j+1} = \{\ell\}$.
By definition of $\ell''$, we know that $\ell''$ is agent~$1$'s highest-valued item in $M_{\pi,j+1}\cup M_{\pi',j+1} = (M_{\pi,j}\cup M_{\pi',j})\setminus\{\ell'\} = M_{\pi,j}$.
Moreover, since $\ell\in (M_{\pi,j+1}\cup M_{\pi',j+1})\subseteq (M_{\pi,j}\cup M_{\pi',j})$, we have that $\ell$ must still be agent~$2$'s highest-valued item in $M_{\pi,j+1}\cup M_{\pi',j+1}$.
It follows that condition~(ii) holds for $j+1$, completing the induction.

To finish the proof of the lemma, observe that for each turn $j\ge i+2$ in which agent~$1$ picks, if condition~(i) holds, the agent receives equal utility in both $\pi$ and $\pi'$, whereas if condition~(ii) holds, the agent receives at least as much utility from $\pi'$ as from $\pi$.
Furthermore, as we mentioned earlier, agent~$1$'s utility from the first $i+1$ turns of $\pi'$ is no less than that of $\pi$.
The desired result follows.
\end{proof}

We now proceed to prove Theorem~\ref{thm:weightmon-2}.

\begin{proof}[Proof of Theorem~\ref{thm:weightmon-2}]
Let $\Pi$ be a weight-consistent family of picking sequences, and consider any $m,\textbf{w}$ such that both $\pi_{2,m,\textbf{w}}$ and $\pi_{2,m,\textbf{w}'}$ belong to $\Pi$, where $\textbf{w}'=(w_1',w_2)$ with $w_1' > w_1$; the assumption that $\textbf{w}'$ and $\textbf{w}$ differ in the first coordinate is without loss of generality.
Since we will not vary the agents or items, we drop $n=2$ and $m$ from the subscripts.
Let $\pi_{\textbf{w}} = (a_1,\dots,a_m)$.
By weight-consistency, $\pi_{\textbf{w}'}$ can be obtained from $\pi_{\textbf{w}}$ by moving some of agent~$1$'s picks earlier, inserting agent~$1$ in some positions, and trimming the suffix of the resulting sequence.
Suppose that the sequence for $\pi_{\textbf{w}'}$ before trimming the suffix has length $m'\ge m$.
We insert agent~$1$'s picks at the end of $\pi_{\textbf{w}}$ so that the sequence also has length $m'$, and add $m'-m$ dummy items for which both agents have utility $0$ such that they are all chosen after the $m$ real items with respect to tie-breaking.
Clearly, the extended picking sequences $\pi_{\textbf{w},\text{ext}}$ and $\pi_{\textbf{w}',\text{ext}}$ with dummy items yield the same utility to both agents as the respective original picking sequences $\pi_{\textbf{w}'}$ and $\pi_{\textbf{w}}$, so we may work with them instead.

We claim that $\pi_{\textbf{w}',\text{ext}}$ can be obtained from $\pi_{\textbf{w},\text{ext}}$ by a series of switches of the following form: for some $1\leq i\leq m'-1$ such that agent~$2$ picks in the $i$-th turn and agent~$1$ in the $(i+1)$st turn of the picking sequence, we switch the position of the two picks.
Indeed, this is immediate when we move agent~$1$'s pick earlier in $\pi_{\textbf{w}}$, whereas inserting agent~$1$'s pick in $\pi_{\textbf{w}}$ is equivalent to moving one of the agent's added picks at the end of $\pi_{\textbf{w},\text{ext}}$ earlier to the appropriate position.
By Lemma~\ref{lem:weightmon-switch}, agent~$1$ receives at least as much utility from $\pi_{\textbf{w}',\text{ext}}$ as from $\pi_{\textbf{w},\text{ext}}$, completing the proof.
\end{proof}

\section{Divisor Methods}
\label{sec:divisor}

As we explained in the introduction, a divisor apportionment method gives rise to a picking sequence that, in each turn, lets an agent $i$ with the smallest $f(t_i)/w_i$ pick the next item (breaking ties in a consistent manner, say, in favor of lower-numbered agents), where $t_i$ denotes the number of times that agent~$i$ has picked so far and $f:\mathbb{Z}_{\ge 0}\rightarrow\mathbb{R}_{\ge 0}$ is a strictly increasing function specific to the method such that $t\leq f(t)\leq t+1$. 
We will refer to the divisor methods and their associated families of picking sequences interchangeably.
By definition, it is clear that every divisor method yields a family of picking sequences (for all $n,m,\textbf{w}$) that are resource-, population-, and weight-consistent.
Theorems~\ref{thm:resmon}, \ref{thm:popmon}, and \ref{thm:weightmon-2} therefore imply the following:

\begin{corollary}
\label{cor:divisor-allmon}
Every divisor method satisfies resource-monotonicity and population-monotonicity; it also satisfies weight-monotonicity when there are two agents.
\end{corollary}

The five traditional divisor methods of Adams, Jefferson, Webster, Hill, and Dean have the function $f(t)$ equal to $t$, $t+1$, $t+\frac{1}{2}$, $\sqrt{t(t+1)}$, and $\frac{t(t+1)}{t+\frac{1}{2}}$, respectively \citep[p.~99]{BalinskiYo01}.
We prove that, perhaps surprisingly, all five methods fail weight-monotonicity in the case of three agents.\footnote{\citet[p.~157]{BramsKa04} showed that for $n=3$, an agent can do worse when her picks move earlier in the picking sequence. However, their example does not correspond to a weight increase with respect to a divisor method and moreover assumes that agents are strategic rather than truthful.}
This also means that weight-consistency does not imply weight-monotonicity beyond two agents.

\begin{proposition}
\label{prop:weightmon-negative}
None of the five traditional divisor methods satisfies weight-monotonicity even when there are three agents.
\end{proposition}

\begin{proof}
First, consider five items and three agents with the following utilities:
\vspace{3mm}
\begin{center}
\begin{tabular}{ c|ccccc } 
  & Item 1 & Item 2 & Item 3 & Item 4 & Item 5 \\
  \hline
 Agent 1 & $10$ & $9$ & $8$ & $7$ & $0$ \\
 Agent 2 & $7$ & $10$ & $8$ & $9$ & $0$ \\
 Agent 3 & $0$ & $7$ & $10$ & $8$ & $9$ \\
\end{tabular}
\end{center}
\vspace{3mm}
If the picking sequence is $\pi_1 = (1,2,1,3,1)$, then agent~$1$ gets items~$1$, $3$, and $4$, and therefore receives a utility of $25$.
On the other hand, if the picking sequence is $\pi_2 = (1,1,2,3,1)$, the agent gets items~$1$, $2$, and $5$, and receives a utility of $19$. 

We claim that for any divisor method whose function $f$ satisfies $f(0) > 0$ and $f(1)^2 > f(0)\cdot f(2)$, there exist $w_1' > w_1 > w_2 > w_3 > 0$ such that the weights $(w_1,w_2,w_3)$ yield the sequence $\pi_1$, while increasing agent~$1$'s weight to $w_1'$ results in the sequence $\pi_2$; the above utility functions then show that the method violates weight-monotonicity.
To prove the claim, let $w_3 = 1$, and
\begin{equation}
\label{eq:weightmon-w2}
1 < w_2 < \frac{f(2)}{f(1)}
\end{equation}
\begin{equation}
\label{eq:weightmon-w1'}
\frac{f(1)}{f(0)}\cdot w_2 < w_1' < \frac{f(2)}{f(0)}
\end{equation}
\begin{equation}
\label{eq:weightmon-w1}
\max\left\{\frac{f(2)}{f(1)}\cdot w_2,\frac{f(1)}{f(0)}\right\} < w_1 < \frac{f(1)}{f(0)}\cdot w_2
\end{equation}
Note that (\ref{eq:weightmon-w1'}) is feasible because $\frac{f(1)}{f(0)}\cdot w_2 < \frac{f(2)}{f(0)}$ by (\ref{eq:weightmon-w2}), and (\ref{eq:weightmon-w1}) is feasible because $f(1)^2 > f(0)\cdot f(2)$.
With the weights $(w_1,w_2,w_3)$, the first pick goes to agent~$1$ because $w_1 > w_2 > w_3$, the second pick goes to agent~$2$ because $\frac{f(0)}{w_2} < \frac{f(1)}{w_1}$ by (\ref{eq:weightmon-w1}), the third pick goes to agent~$1$ because $\frac{f(0)}{w_3} = f(0) > \frac{f(1)}{w_1}$ by (\ref{eq:weightmon-w1}), the fourth pick goes to agent~$3$ because $\frac{f(0)}{w_3} = f(0) < \frac{f(2)}{w_1'} < \frac{f(2)}{w_1} < \frac{f(1)}{w_2}$ by (\ref{eq:weightmon-w1'}) and (\ref{eq:weightmon-w1}), and the fifth pick goes to agent~$1$ because $\frac{f(2)}{w_1} < \frac{f(1)}{w_2}$ by (\ref{eq:weightmon-w1}).
Similarly, with the weights $(w_1',w_2,w_3)$, the first pick goes to agent~$1$ because $w_1' > w_2 > w_3$, the second pick goes to agent~$1$ because $\frac{f(1)}{w_1'} < \frac{f(0)}{w_2}$ by (\ref{eq:weightmon-w1'}), the third pick goes to agent~$2$ because $\frac{f(2)}{w_1'} > f(0) > \frac{f(0)}{w_2}$ by (\ref{eq:weightmon-w1'}), the fourth pick goes to agent~$3$ because $\frac{f(0)}{w_3} = f(0) < \frac{f(2)}{w_1'} < \frac{f(2)}{w_1} < \frac{f(1)}{w_2}$ by (\ref{eq:weightmon-w1'}) and (\ref{eq:weightmon-w1}), and the fifth pick goes to agent~$1$ because $\frac{f(2)}{w_1'} < \frac{f(2)}{w_1} < \frac{f(1)}{w_2}$ by (\ref{eq:weightmon-w1}).
This establishes the claim.
Since the conditions $f(0) > 0$ and $f(1)^2 > f(0)\cdot f(2)$ are met for Jefferson's and Webster's methods, as well as any method with $f(t)=t+c$ for some constant $c > 0$, these methods do not satisfy weight-monotonicity for three agents.

Next, we show that any divisor method whose function $f$ satisfies $f(0) = 0$ and $f(2)^2 > f(3)\cdot f(1)$ also fails weight-monotonicity.
To this end, we add three extra items to the above instance---the extra items are of value $100$, $99$, and $98$ to all agents, respectively. 
Since $f(0) = 0$, with lexicographic tie-breaking, the first three picks are $(1,2,3)$, and these agents pick the extra items in decreasing order of utility.
The rest of the picking sequence proceeds as in the previous instance, only with the number of items picked shifted upwards by one.
Hence, we may select $w_1',w_1,w_2,w_3$ so that the remaining picks are $(1, 2, 1, 3, 1)$ for the weights $(w_1,w_2,w_3)$ and $(1, 1, 2, 3, 1)$ for the weights $(w_1',w_2,w_3)$, again violating weight-monotonicity.
Finally, one can verify that the conditions $f(0)=0$ and $f(2)^2 > f(3)\cdot f(1)$ are indeed met for Adams', Hill's, and Dean's methods.
\end{proof}

We now explore how divisor methods fare with respect to the three fairness notions.
We start by showing that Adams' method is the unique option if WEF1 is desired.

\begin{theorem}
\label{thm:divisor-WEF1}
Among all divisor methods, Adams' method is the only one satisfying WEF1.
\end{theorem}

\begin{proof}
\citet{ChakrabortyIgSu20} already showed that Adams' method satisfies WEF1.
To establish uniqueness, consider a divisor method with function $f$ that satisfies WEF1.
We first claim that $f(0) = 0$; suppose for contradiction that $f(0) > 0$.
Consider an instance where $m=n$, every agent has a positive utility for every item, and $w_n > \frac{f(1)}{f(0)}\cdot w_i$ for all $i\in[n-1]$.
With these weights, the first two picks go to agent~$n$.
Since the number of items is equal to the number of agents, there exists an agent who does not receive any item.
Such an agent will (weighted-)envy agent~$n$ by more than one item, a contradiction.
Hence $f(0)=0$.

Next, assume for contradiction that $\frac{f(a)}{f(b)} < \frac{a}{b}$ for some $1\le b < a$.
Consider an instance with $a+b+1$ items and two agents such that $\frac{f(a)}{f(b)} < \frac{w_1}{w_2} < \frac{a}{b}$.
Assume that agent~$2$ has utility $1$ for every item.
Since $\frac{f(a)}{w_1} < \frac{f(b)}{w_2}$,  agent~$2$ picks at most $b$ times while agent~$1$ picks at least $a+1$ times.
Hence, for any $B\subseteq M_1$ with $|B|\le 1$, we have
\[
\frac{u_2(M_2)}{w_2} \le \frac{b}{w_2} < \frac{a}{w_1} \le \frac{u_2(M_1\setminus B)}{w_1},
\]
contradicting the assumption that the divisor method satisfies WEF1.
It follows that  $\frac{f(a)}{f(b)} \ge \frac{a}{b}$ for all $1\le b < a$.

We are now ready to show that $f(t)=t$ for all $t\in\mathbb{Z}_{\ge 0}$, which suffices to complete the proof.
The case $t=0$ was already handled in the first paragraph.
If $f(t) > t$ for some $t$, say $f(t) = t+\delta$ with $\delta > 0$, then the property established in the previous paragraph implies that $f(kt) \ge kf(t) = kt + k\delta$ for every positive integer $k$.
When $k > 1/\delta$, we have $f(kt)\ge kt + k\delta > kt+1$, violating the definition of divisor methods.
Hence $f(t)=t$ for all $t$, as desired.
\end{proof}

Next, we show that interestingly, a broad class of divisor methods, which includes all five traditional divisor methods, satisfies WWEF1.
To this end, we first derive a characterization of divisor methods fulfilling WWEF1.

\begin{theorem}
\label{thm:divisor-WWEF1-char}
A divisor method with function $f$ satisfies WWEF1 if and only if both of the following conditions hold:
\begin{itemize}
\item $\frac{f(a)}{f(b)}\le\frac{a}{b}$ for any integers $1\le b\le a$;
\item $\frac{f(a)}{f(b)}\le\frac{a+1}{b+1}$ for any integers $0\le a\le b$ with $b\ne 0$.
\end{itemize}
\end{theorem}

\begin{proof}
($\Rightarrow$) We begin with the forward direction. 
Consider a divisor method satisfying WWEF1, and assume for contradiction that (at least) one of the two conditions is violated.
\begin{itemize}
\item Suppose that $\frac{f(a)}{f(b)}>\frac{a}{b}$ for some $1\le b\le a$.
Consider an instance with $a+b+1$ items and two agents such that $\frac{f(a)}{f(b)} > \frac{w_1}{w_2} > \frac{a}{b} \ge 1$.
Since $\frac{f(a)}{w_1} > \frac{f(b)}{w_2}$, agent~$1$ picks at most $a$ times while agent~$2$ picks at least $b+1 \ge 2$ times.
The relation $\frac{w_1}{w_2} > \frac{a}{b}$ implies that the first condition of Theorem~\ref{thm:WWEF1} is violated when $i=1$ and $j=2$, a contradiction.
\item Suppose that $\frac{f(a)}{f(b)} > \frac{a+1}{b+1}$ for some $0\le a\le b$ with $b\ne 0$.
Consider an instance with $a+b+1$ items and two agents such that $\frac{f(a)}{f(b)} > \frac{w_1}{w_2} > \frac{a+1}{b+1}$.
Since $\frac{f(a)}{w_1} > \frac{f(b)}{w_2}$, agent~$1$ picks at most $a$ times while agent~$2$ picks at least $b+1 \ge 2$ times.
If $w_1 \le w_2$, the relation $\frac{w_1}{w_2} > \frac{a+1}{b+1}$ implies that the second condition of Theorem~\ref{thm:WWEF1} is violated when $i=1$ and $j=2$, a contradiction.
Else, $w_1 > w_2$, and so $\frac{a}{b} \le 1 < \frac{w_1}{w_2}$, contradicting the first condition of Theorem~\ref{thm:WWEF1}.
\end{itemize}

($\Leftarrow$) We now proceed to the backward direction.
Consider a divisor method whose function~$f$ satisfies both conditions in the theorem statement, and fix a pair of agents $i,j$.
It suffices to show that every time agent~$j$ picks an item starting from her second pick, the conditions in Theorem~\ref{thm:WWEF1} are satisfied. 
By definition of the divisor method, after agent $j$'s pick it holds that $\frac{f(t_j-1)}{w_j}\le\frac{f(t_i)}{w_i}$; otherwise agent~$i$ should have picked instead of agent~$j$.
Consider two cases as in the conditions of Theorem~\ref{thm:WWEF1} with $t_j\ge 2$.

\underline{Case 1}: $w_i\ge w_j$.
Since $f$ is strictly increasing, we have $f(t_j-1) > f(0) \ge 0$, and so $1\le\frac{w_i}{w_j}\le\frac{f(t_i)}{f(t_j-1)}$.
This means that $t_i\ge t_j-1$, and our assumption on $f$ implies that $\frac{w_i}{w_j}\le\frac{f(t_i)}{f(t_j-1)} \le \frac{t_i}{t_j-1}$, as desired.

\underline{Case 2}:
$w_i\le w_j$.
We have $f(t_j-1) > 0$ and $\frac{w_i}{w_j}\le\frac{f(t_i)}{f(t_j-1)}$.
If $t_i\ge t_j-1$, then  $\frac{w_i}{w_j}\le 1\le\frac{t_i+1}{t_j}$.
Otherwise, $t_i < t_j-1$, and our assumption on $f$ implies that $\frac{w_i}{w_j}\le\frac{f(t_i)}{f(t_j-1)} \le \frac{t_i+1}{(t_j-1)+1} = \frac{t_i+1}{t_j}$, as desired.
\end{proof}

The characterization in \Cref{thm:divisor-WWEF1-char} involves two variables.
We now present an alternative characterization that involves only a single variable.

\begin{corollary}
\label{cor:divisor-WWEF1-char-alt}
A divisor method with function $f$ satisfies WWEF1 if and only if 
\[
\frac{t}{t+1} \le \frac{f(t)}{f(t+1)} \le \frac{t+1}{t+2}
\]
for all integers $t\ge 0$.
\end{corollary}

\begin{proof}
We first show that our condition is necessary for WWEF1.
The right inequality follows from the second condition of \Cref{thm:divisor-WWEF1-char} by taking $a=t$ and $b=t+1$.
The left inequality holds trivially for any divisor method when $t=0$, while for $t\ge 1$, it follows from the first condition of \Cref{thm:divisor-WWEF1-char} by taking $a=t+1$ and $b=t$.

Next, we show that our condition is sufficient for WWEF1 by showing that it implies the two conditions of \Cref{thm:divisor-WWEF1-char} when $a\ne b$ (the case $a=b$ is trivial).
First, for $1\le b< a$, our left inequality implies that
\begin{align*}
\frac{f(a)}{f(b)} 
= \frac{f(b+1)}{f(b)}\cdot \frac{f(b+2)}{f(b+1)} \cdot \cdots \cdot \frac{f(a)}{f(a-1)} \le \frac{b+1}{b}\cdot \frac{b+2}{b+1} \cdot \cdots \cdot \frac{a}{a-1} = \frac{a}{b}.
\end{align*}
Second, for $0\le a< b$, our right inequality implies that
\begin{align*}
\frac{f(a)}{f(b)} 
= \frac{f(a)}{f(a+1)}\cdot \frac{f(a+1)}{f(a+2)} \cdot \cdots \cdot \frac{f(b-1)}{f(b)} \le \frac{a+1}{a+2}\cdot \frac{a+2}{a+3} \cdot \cdots \cdot \frac{b}{b+1} = \frac{a+1}{b+1}.
\end{align*}
This completes the proof.
\end{proof}

Not all divisor methods satisfy WWEF1, for example:
\[
f(t) = \begin{cases}
    t & \text{for } 0 \leq t \leq 1; \\
    t+1 & \text{for } t \ge 2.
  \end{cases}
\]
Indeed, this can be readily verified by taking $t=1$ in \Cref{cor:divisor-WWEF1-char-alt}.
Nevertheless, we show next that a large class of divisor methods (including all five traditional divisor methods)
satisfies WWEF1, meaning that any method from this class can guarantee fairness beyond the setting of identical items (i.e., apportionment).
Specifically, for any $p\in\mathbb{R}$ and $w\in[0,1]$, consider the divisor method defined by
\[
f_{p,w}(t) = \begin{cases}
    \left(w\cdot t^p + (1-w)\cdot (t+1)^p\right)^{1/p} & \text{for } p \neq 0; \\
    t^w(t+1)^{1-w} & \text{for } p = 0; \\
  \end{cases}
\]
When $t=0$ and $p \le 0$, we define $f_{p,w}(0) = 0$.
Note that the expression for $p = 0$ can also be derived by taking the limit of the general formula.
When $p=1$, we obtain the class of \emph{stationary divisor methods}, where $f(t) = t+c$ for some constant $c\in[0,1]$.
When $w=1/2$, we obtain the class of \emph{power-mean divisor methods} \citep[p.~68]{Pukelsheim14}.
We will refer to our class of divisor methods as the \emph{weighted power-mean divisor methods}.
Notice that this class includes all five traditional divisor methods, which correspond to taking $(p,w)=(1,1)$, $(-1,1/2)$, $(0,1/2)$, $(1,1/2)$, and $(1,0)$ for Adams, Dean, Hill, Webster, and Jefferson, respectively.

\begin{theorem}
\label{thm:divisor-WWEF1}
All weighted power-mean divisor methods satisfy WWEF1.
\end{theorem}

\begin{proof}
It suffices to show that the functions $f_{p,w}$ of all weighted power-mean divisor methods satisfy the two conditions of \Cref{thm:divisor-WWEF1-char}.

First, assume that $p\neq 0$.
If $1\le b\le a$, the condition $\frac{f_{p,w}(a)}{f_{p,w}(b)} \le \frac{a}{b}$ can be rewritten as
\begin{equation}
\label{eq:weighted-power-mean}
\left(\frac{w\cdot a^p + (1-w)\cdot (a+1)^p}{w\cdot b^p + (1-w)\cdot (b+1)^p}\right)^{1/p} \le \frac{a}{b}.
\end{equation}
For $p > 0$, this is equivalent to 
\[
\frac{w\cdot a^p + (1-w)\cdot (a+1)^p}{w\cdot b^p + (1-w)\cdot (b+1)^p} \le \left(\frac{a}{b}\right)^p,
\]
which is equivalent to $\left(\frac{a+1}{b+1}\right)^p \le \left(\frac{a}{b}\right)^p$, or $\frac{a+1}{b+1}\le\frac{a}{b}$; the latter holds since $a\ge b$.
On the other hand, for $p < 0$, the inequality~\eqref{eq:weighted-power-mean} is equivalent to 
\[
\frac{w\cdot a^p + (1-w)\cdot (a+1)^p}{w\cdot b^p + (1-w)\cdot (b+1)^p} \ge \left(\frac{a}{b}\right)^p,
\]
which is equivalent to $\left(\frac{a+1}{b+1}\right)^p \ge \left(\frac{a}{b}\right)^p$, or $\frac{a+1}{b+1}\le\frac{a}{b}$; the latter again holds since $a\ge b$.
This shows that the first condition of \Cref{thm:divisor-WWEF1-char} is fulfilled.
The second condition can be established in an analogous manner---at the end we obtain the opposite inequality $\frac{a}{b} \le \frac{a+1}{b+1}$, which holds because $a\le b$.

Next, consider $p=0$.
If $1\le b\le a$, the condition $\frac{f_{p,w}(a)}{f_{p,w}(b)} \le \frac{a}{b}$ can be rewritten as
\[
\frac{a^w(a+1)^{1-w}}{b^w(b+1)^{1-w}} \le \frac{a}{b}.
\]
This is equivalent to $\frac{a+1}{b+1}\le\frac{a}{b}$, which holds since $a\ge b$.
Consider now $0\le a\le b$ with $b\ne 0$.
If $a=0$, then since we define $f_{0,w}(0) = 0$, the condition $\frac{f_{p,w}(a)}{f_{p,w}(b)} \le \frac{a+1}{b+1}$ holds trivially.
For $a\ge 1$, this condition can be rewritten as
\[
\frac{a^w(a+1)^{1-w}}{b^w(b+1)^{1-w}} \le \frac{a+1}{b+1},
\]
which is equivalent to $\frac{a}{b}\le\frac{a+1}{b+1}$; the latter holds since $a\le b$.
\end{proof}

The following corollary is an immediate consequence of \Cref{thm:divisor-WWEF1} and the remarks preceding it.

\begin{corollary}
\label{cor:divisor-WWEF1-traditional}
All five traditional divisor methods satisfy WWEF1.
\end{corollary}

To see that the class of divisor methods satisfying WWEF1 is strictly larger than that of weighted power-mean methods, consider
\[
f(t) = \begin{cases}
    t+\frac{1}{2} & \text{for } 0 \leq t \leq 1; \\
    t+1 & \text{for } t \ge 2.
  \end{cases}
\]
Indeed, for weighted power-mean methods, we have $f_{p,0}(t) = t+1$ for all $p,t$, and $f_{p,w}(t) < t+1$ for all $w\in (0,1]$ and all $p,t$, so none of these methods coincides with $f$.
The fact that $f$ satisfies WWEF1 can be readily verified using \Cref{cor:divisor-WWEF1-char-alt}.

Finally, we turn to WPROP1, where we illustrate a strong relationship with a notion from the apportionment setting.
A picking sequence $\pi_{n,m,\textbf{w}}$ is said to satisfy \emph{lower quota} if for any $i\in N$, it holds that $t_i\ge\left\lfloor\frac{w_i\cdot m}{\sum_{i'\in N}w_{i'}}\right\rfloor$, where $t_i$ denotes the number of picks in $\pi_{n,m,\textbf{w}}$ assigned to agent~$i$.

\begin{proposition}
\label{prop:lower-quota-WPROP1}
Let $\pi$ be a picking sequence such that every prefix of $\pi$ satisfies lower quota. 
Then $\pi$ satisfies WPROP1.
\end{proposition}

\begin{proof}
For any prefix of such a picking sequence $\pi$ and any $i\in N$, we have
$t_i\ge \left\lfloor\frac{w_i\cdot k}{\sum_{i'\in N}w_{i'}}\right\rfloor > \frac{w_i\cdot k}{\sum_{i'\in N}w_{i'}} - 1$, where $t_i$ and $k$ denote the number of agent~$i$'s picks in the prefix and the length of the prefix, respectively.
Theorem~\ref{thm:WPROP1} implies that $\pi$ satisfies WPROP1.
\end{proof}

Since Jefferson's method satisfies lower quota \citep[p.~130]{BalinskiYo01} and is resource-consistent, any prefix of its associated picking sequence also satisfies lower quota.
By Proposition~\ref{prop:lower-quota-WPROP1}, the method satisfies WPROP1.
We prove that it is the only divisor method to do so.

\begin{theorem}
\label{thm:divisor-WPROP1}
Among all divisor methods, Jefferson's method is the only one satisfying WPROP1.
\end{theorem}

\begin{proof}
As explained before the theorem statement, Jefferson's method satisfies WPROP1.

To establish uniqueness, consider a divisor method with function $f$ that satisfies WPROP1.
We first claim that $f(0) > 0$; suppose for contradiction that $f(0) = 0$.
Consider an instance where $m=n\ge 3$ and agent~$1$ has utility~$1$ for every item.
Regardless of the weights, each agent will receive one item.
Assume that $2/n < w_1 < 1$ and $w_2=w_3=\dots=w_n = \frac{1-w_1}{n-1}$.
Then, for any $B\subseteq M\setminus M_1$ with $|B|\le 1$, we have
\[
u_1(M_1) = 1 < w_1\cdot n - 1 \le \frac{w_1}{w_1+\dots+w_n}\cdot u_1(M) - u_1(B),
\]
meaning that the method fails WPROP1, a contradiction.
Hence $f(0) > 0$.

Next, we claim that $\frac{f(a)}{f(b)} \le \frac{a+1}{b+1}$ for all $0\le b < a$.
Assume for contradiction that $\frac{f(a)}{f(b)} > \frac{a+1}{b+1}$ for some $0\le b < a$.
Choose $w_1,w_2$ such that $\frac{f(a)}{f(b)} > \frac{w_1}{w_2} > \frac{a+1}{b+1}$, choose $n$ large enough so that $\frac{w_1}{w_2} > \frac{a+1}{b+1-\frac{1}{n-1}}$, and consider an instance with $a+(n-1)(b+1)$ items and $n$ agents such that $w_2 = w_3 = \dots = w_n$.
Assume that agent~$1$ has utility $1$ for every item.
Since $\frac{f(a)}{w_1} > \frac{f(b)}{w_2}$, agent~$1$ picks at most $a$ times---indeed, if agent~$1$ picks at least $a+1$ times, then one of the remaining agents will pick at most $b$ times, which is impossible due to the inequality.
Moreover, since $\frac{w_1}{w_2} > \frac{a+1}{b+1-\frac{1}{n-1}}$, we have
\[
w_1+(a+1)(n-1)w_2 < w_1(n-1)(b+1),
\]
which implies that
\[
(a+1)(w_1+(n-1)w_2) < w_1(a+(n-1)(b+1)),
\]
or equivalently 
\[
a < \frac{w_1}{w_1+(n-1)w_2}\cdot(a+(n-1)(b+1))-1.
\]
Hence, for any $B\subseteq M\setminus M_1$ with $|B|\le 1$, it holds that
\begin{align*}
u_1(M_1) \le a &< \frac{w_1}{w_1+(n-1)w_2}\cdot (a+(n-1)(b+1)) - 1 \\
&\le \frac{w_1}{w_1+\dots+w_n}\cdot u_1(M) - u_1(B).
\end{align*}
This shows that the divisor method fails WPROP1, a contradiction.
It follows that $\frac{f(a)}{f(b)} \le \frac{a+1}{b+1}$ for all $0\le b < a$.

We are now ready to show that $f(t)=t+1$ for all $t\in\mathbb{Z}_{\ge 0}$, which suffices to complete the proof.
Suppose for contradiction that $f(t) = t+1-\delta$ with $\delta > 0$ for some $t$.
If $t > 0$, the property established in the previous paragraph implies that $f(kt) \le \frac{kt+1}{t+1}\cdot f(t) = \frac{kt+1}{t+1}\cdot(t+1-\delta)$ for every positive integer $k$; when $k > \frac{t+1-\delta}{\delta t}$, we have $f(kt) \le \frac{kt+1}{t+1}\cdot(t+1-\delta) < kt$, violating the definition of divisor methods. 
Else, $t=0$, and the property in the previous paragraph implies that $f(1) \le 2f(0) < 2$, so we can continue as in the $t > 0$ case by taking $t=1$ and again obtain a contradiction.
Hence $f(t)=t+1$ for all $t\in\mathbb{Z}_{\ge 0}$, as desired.
\end{proof}

\section{Quota Method}
\label{sec:quota}

Although divisor methods are widely used in practice, they do come with an axiomatic downside: no divisor method satisfies an arguably natural axiom known as \emph{quota} \citep[p.~130]{BalinskiYo01}. 
A picking sequence satisfies the quota axiom if for every $i \in N$, it holds that $\left\lfloor \frac{w_i \cdot m }{\sum_{i' \in N}w_{i'}} \right\rfloor \leq  t_i \leq \left\lceil \frac{w_i \cdot m}{\sum_{i' \in N}w_{i'}} \right\rceil,$ where $t_i$ is the number of picks assigned to agent $i$ by the picking sequence---note that the lower bound simply corresponds to the lower quota notion introduced before Proposition~\ref{prop:lower-quota-WPROP1}.
Motivated by this observation, \cite{BalinskiYo75} proposed the \emph{quota method} which satisfies the quota axiom as well as resource-consistency.
Intuitively, this method can be seen as a constrained version of Jefferson's method where we choose an agent $i$ minimizing $(t_i+1)/w_i$ over a restricted subset of ``eligible'' agents.
The picking sequence for the quota method is determined iteratively.
For each round $k \in [m]$, let $t_i$ be the number of times agent~$i$ has picked in rounds $1,\dots, k-1$. 
An agent is \emph{eligible} if she would not exceed her upper bound in the quota axiom upon getting an additional pick in round $k$. 
Equivalently, the set of eligible agents is $U(\textbf{w},\textbf{t},k) = \left\{i \in N \,\middle|\, t_i < \frac{w_i \cdot k }{ \sum_{i' \in N} w_{i'}}\right\}$, where $\textbf{t}=(t_1,\dots,t_n)$. 
Among all eligible agents, the next pick is assigned to an agent minimizing $(t_i+1)/w_i$, breaking ties in a consistent manner. 
The method trivially satisfies resource-consistency, which by Theorem~\ref{thm:resmon} implies the following:

\begin{corollary}
The quota method satisfies resource-monotonicity. 
\end{corollary}

However, satisfying the quota axiom comes at a price: 
in contrast to all divisor methods (Corollary~\ref{cor:divisor-allmon}), the quota method fails population-monotonicity.
Moreover, like the five traditional divisor methods (Proposition~\ref{prop:weightmon-negative}), the quota method fails weight-monotonicity for $n=3$. 

\begin{proposition}
\label{prop:quota-popmon}
The quota method does not satisfy population-monotonicity. 
In addition, it does not satisfy weight-monotonicity even when there are three agents. 
\end{proposition}

\begin{proof}
We begin with population-monotonicity.
First, consider the first four agents and all three items in the following instance: 
\vspace{3mm}
\begin{center}
\begin{tabular}{ c|ccc } 
  & Item 1 & Item 2 & Item 3  \\
  \hline
 Agent 1 & $2$ & $1$ & $0$  \\
 Agents 2, 3, 4 & $0$ & $1$ & $0$ \\ \hdashline
 Agent 5 & $0$ & $0$ & $1$ \\
\end{tabular}
\end{center}
\vspace{3mm}
Let $w_1 = 1/2$ and $w_2 = w_3 = w_4 = 1/6$, and note that $\sum_{i=1}^4w_i = 1$. 
We claim that the picking sequence selected by the quota method is $(1,2,1)$. 
Indeed, in the first round all four agents are eligible and agent $1$ uniquely minimizes $1/w_i$. 
In the second round, all agents except agent $1$ are eligible and agent $2$ is selected by tie-breaking. 
In the last round, the set of eligible agents is $\{1,3,4\}$ and agent $1$ is selected since $4 = (t_1  + 1)/w_1< (t_3 + 1)/w_3 = 6$. The picking sequence $(1,2,1)$ induces a utility of $2$ for agent $1$. 

In the second instance, agent $5$ with $w_5 = 1/3$ is added to the previous instance. 
The weights of the agents now sum up to $4/3$. 
We claim that the picking sequence selected by the quota method in the modified example is $(1,5,1)$. 
Indeed, in the first round, all agents are eligible and agent $1$ is selected as it uniquely minimizes $1/w_i$. 
In the second round, all agents but agent $1$ are eligible and agent $5$ minimizes $1/w_i$. 
In the last round, the set of eligible agents is $\{1,2,3,4\}$ and agent $1$ is selected again since $4=(t_1 + 1)/w_1 < (t_2 + 1)/w_2 = 6$. 
Since the picking sequence $(1,5,1)$ induces a utility of $3$ for agent $1$, the example shows that the arrival of agent $5$ induced an increase in the utility of agent $1$. 
Hence, the quota method fails population-monotonicity. 

Next, we address weight-monotonicity.
Recall the example with five items and three agents from the proof of Proposition~\ref{prop:weightmon-negative}. 
We will use the same example to demonstrate that the quota method fails weight-monotonicity. 
To this end, we construct weights $w'_1 > w_1 > w_2 > w_3$ such that the picking sequences $\pi_1 = (1,2,1,3,1)$ and $\pi_2=(1,1,2,3,1)$ are selected by the quota method for the instances with weights $(w_1,w_2,w_3)$ and $(w'_1,w_2,w_3)$, respectively.
As we argued in Proposition~\ref{prop:weightmon-negative}, $\pi_2$ induces a smaller utility than $\pi_1$ for agent~$1$, yielding the desired violation of weight-monotonicity.

The weights of the first instance are $w_1=9/18$, $w_2=5/18$ and $w_3=4/18$; we go through the five rounds of the quota method to show that $\pi_1$ is selected. 
In the first round, all three agents are eligible and agent~$1$  uniquely minimizes $1/w_i$. 
In the second round, all but agent~$1$ are eligible and agent~$2$ minimizes $1/w_i$, so the second pick goes to agent $2$. In the third round, the set of eligible agents is $\{1,3\}$ and since $2/w_1 = 4 < 9/2 = 1/w_3$, agent~$1$ gets another pick. 
In the fourth round, agents~$2$ and $3$ are eligible and since $1/w_3 < 2/w_2$, agent $3$ is selected in this round. 
Finally, all agents are eligible and since $3/w_1 = 6 < 36/5=2/w_2$, the last pick goes to agent $1$.  
Hence $\pi_1$ is selected.

For the modified instance we increase the weight of agent~$1$ to $w'_1=11/18$; we show that the quota method selects $\pi_2$. 
In the first round, all agents are eligible and agent~$1$ uniquely minimizes $1/w_i$. 
Contrary to the first instance, agent~$1$ is still eligible for the second round; moreover, since $2/w_1<1/w_2$, the second pick also goes to agent~$1$.
In the third round, agents~$2$ and $3$ are eligible and since $1/w_2<1/w_3$, the third pick goes to agent~$2$. 
In the fourth round, agents~$1$ and $3$ are eligible and since $1/w_3<3/w_1$, the pick is assigned to agent~$3$. 
Finally, the set of eligible agents is $\{1,2\}$ and as $3/w_1<2/w_2$, the last pick goes to agent~$1$ again, meaning that $\pi_2$ is indeed chosen.
\end{proof}

As we observed in Section~\ref{sec:divisor}, all divisor methods are weight-consistent for any number of agents by definition, and therefore weight-monotone for two agents by Theorem~\ref{thm:weightmon-2}.
In contrast, we show in Proposition~\ref{prop:quota-weights} that the quota method is not weight-consistent.
However, for two agents, we prove that the method is weight-consistent and hence weight-monotone.

\begin{theorem}
\label{thm:quota-weight-consistency}
The quota method satisfies weight-consistency and weight-monotonicity when there are two agents.
\end{theorem}

\begin{proof}
By Theorem~\ref{thm:weightmon-2}, it suffices to establish weight-consistency.
Let $\Pi$ be the family of picking sequences induced by the quota method and consider any $m,\textbf{w}$ such that both $\pi_{2,m,\textbf{w}}$ and $\pi_{2,m,\textbf{w}'}$ belong to $\Pi$, where $\textbf{w}' = (w'_1,w_2)$ with $w'_1 > w_1$; the assumption that $\textbf{w}'$ and $\textbf{w}$ differ in the first coordinate is without loss of generality. 
Since we will not vary the agents or items, we drop $n=2$ and $m$ from the subscripts. 
The proof of the theorem proceeds in two steps. 
In the first step, we show that for any two prefixes of the same length of $\pi_{\textbf{w}}$ and  $\pi_{\textbf{w}'}$, the number of picks that agent $1$ gets in the latter is at least as large as that in the former. 
In the second step, we show that this property suffices to guarantee that $\pi_{\textbf{w}'}$ can be constructed from $\pi_{\textbf{w}}$ by the type of modifications that are specified within the definition of weight-consistency. 
Hence, the two steps together complete the proof. 

For the first step, we show by induction that for any $k \in [m]$, the prefix of length $k$ of $\pi_{\textbf{w}'}$ contains at least as many picks of agent $1$ as the corresponding prefix of $\pi_{\textbf{w}}$. 
For the base case $k=1$, observe that both agents are eligible in the first round, and the method simply selects the agent with maximum weight with consistent tie-breaking. 
Hence, it is not possible that agent $1$ gets the first pick in $\pi_{\textbf{w}}$ while agent $2$ gets the first pick in $\pi_{\textbf{w}'}$. 
For the induction step, fix $k \in [m-1]$ and assume that the statement is true for all prefixes of length at most $k$; we show that the statement also holds for $k+1$. 
If the number of agent $1$'s picks in the $k$-length prefix of $\pi_{\textbf{w}'}$ is strictly greater than that in $\pi_{\textbf{w}}$, then the statement is true for the prefixes of length $k+1$ regardless of the $(k+1)$st picks. 
Hence, assume that these two numbers are equal, and call this quantity $t_1$.
Similarly, we define $t_2 = k - t_1$ to be the number of picks that agent~$2$ gets in the $k$-length prefix of $\pi_{\textbf{w}}$ (and $\pi_{\textbf{w}'}$). 
Observe that if agent~$1$ is eligible in the $(k+1)$st round of the original instance, then she is also eligible in the $(k+1)$st round of the modified instance---indeed, this is because $t_1 < \frac{w_1}{w_1 + w_2}\cdot(k+1)$ implies $t_1 < \frac{w'_1}{w'_1 + w_2}\cdot(k+1)$.
Conversely, if agent $2$ is eligible in the modified instance, then she is also eligible in the original instance---this is because $t_2 < \frac{w_2}{w'_1 + w_2}\cdot(k+1)$ implies $t_2 < \frac{w_2}{w_1 + w_2}\cdot(k+1)$. 
We consider a case distinction. 
If agent~$2$ is the only eligible agent in the original instance, then the desired property is satisfied no matter which agent is selected in the modified instance. 
If agent~$1$ is the only eligible agent in the original instance, then, by the above claims she is the only eligible agent in the modified instance and is hence selected. 
In the last case where both agents are eligible in the original instance, we have that either $\{1\}$ or $\{1,2\}$ is the set of eligible agents in the modified instance. 
While the former case is trivial again, consider the latter case. 
If agent~$1$ is selected in the original instance, we have that $(t_1 + 1)/w'_1 < (t_1 + 1) /w_1 \leq (t_2 + 1)/w_2$, and so agent~$1$ is also selected in the modified instance. 
This concludes the induction and the first step of the proof.   

We start the second step by defining $T_1$ and $T'_1$ as the number of times that agent~$1$ is allowed to pick in the entire sequence $\pi_{\textbf{w}}$ and $\pi_{\textbf{w}'}$, respectively. 
We show that we can modify $\pi_{\textbf{w}}$ by shifting turns of agent~$1$ to the front, inserting turns of agent $1$ and trimming, and arrive at the picking sequence $\pi_{\textbf{w}'}$. 
We start by inserting $T'_1 - T_1$ turns of agent $1$ at the end of $\pi_{\textbf{w}}$. 
Now, by the property of the prefixes which we derived in the first part of the proof, there exists a perfect matching between the positions of agent~$1$'s picks in $\pi_{\textbf{w}}$ and the positions of agent~$1$'s picks in $\pi_{\textbf{w}'}$ such that for each of the agent's pick in $\pi_{\textbf{w}}$, the corresponding pick in $\pi_{\textbf{w}'}$ does not occur later.
We now modify $\pi_{\textbf{w}}$ by going through the sequence from left to right and moving every pick of agent $1$ to its ``partner position'' from the perfect matching. 
By the described property, we only move picks towards the front of the sequence. 
Lastly, we trim $\pi_{\textbf{w}}$ so that it is of length $m$ again. 
Since all positions which are not filled with a pick of agent $1$ are automatically filled with a pick of agent $2$, $\pi_{\textbf{w}}$ now coincides with $\pi_{\textbf{w}'}$. 
This concludes the proof. 
\end{proof}

Next, we address fairness criteria for the quota method.

\begin{theorem}
\label{thm:quota-fairness}
The quota method fails WEF1 but satisfies WWEF1 and WPROP1.
\end{theorem}

\begin{proof}
Since the quota method satisfies lower quota, Proposition~\ref{prop:lower-quota-WPROP1} implies that it satisfies WPROP1.
On the other hand, no rule can simultaneously satisfy WPROP1 and WEF1 \citep{ChakrabortyIgSu20}, so the method fails WEF1.

We now show that the quota method satisfies WWEF1.
Fix a pair of agents $i,j$.
It suffices to show that every time $j$ picks an item starting from her second pick, the conditions in Theorem~\ref{thm:WWEF1} are fulfilled. 
If $i$ was eligible just before $j$'s pick, then after $j$'s pick it holds that $\frac{t_j}{w_j}\le\frac{t_i+1}{w_i}$.
In this case, we can proceed as in the backward direction of Theorem~\ref{thm:divisor-WWEF1-char} with $f(t) = t+1$.
Suppose therefore that $i$ was not eligible just before $j$'s pick, and this pick is the $k$-th pick of the sequence.
By the eligibility criterion, we have $t_i\ge \frac{w_i\cdot k}{\sum_{i'\in N}w_{i'}}$.
On the contrary, $j$ was eligible for this pick, so $t_j-1 <  \frac{w_j\cdot k}{\sum_{i'\in N}w_{i'}}$.
Consider two cases as in the conditions of Theorem~\ref{thm:WWEF1} with $t_j\ge 2$.

\underline{Case~1}: $w_i\ge w_j$.
Dividing the first inequality above by the second immediately yields $\frac{w_i}{w_j}\le\frac{t_i}{t_j-1}$.

\underline{Case~2}: $w_i\le w_j$.
If $t_i\ge t_j-1$, then $\frac{w_i}{w_j}\le 1\le\frac{t_i+1}{t_j}$.
Otherwise, $t_i < t_j-1$, and so $\frac{w_i}{w_j}\le\frac{t_i}{t_j-1}\le\frac{t_i+1}{t_j}$, where the first inequality is derived as in Case~1.
\end{proof}

\section{Maximum (Weighted) Nash Welfare}
\label{sec:MNW}

In this section, we turn our attention from picking sequences to another important fair allocation rule.
Given any instance in the unweighted setting, the \emph{maximum Nash welfare (MNW)} solution chooses an allocation that maximizes the \emph{Nash welfare}, i.e., the product of the agents' utilities.
MNW is known to satisfy strong fairness guarantees including EF1 \citep{CaragiannisKuMo19,HalpernPrPs20}.
When weights are present, a natural generalization called \emph{maximum weighted Nash welfare (MWNW)}, which maximizes\footnote{Ties can be broken arbitrarily unless the maximum weighted Nash welfare is $0$ (which occurs, for example, when $m<n$). In this exceptional case, we choose a maximum subset of agents who can be given positive utilities simultaneously, breaking ties in a consistent manner among all such subsets independently of the weights (e.g., lexicographically with respect to the agent indices). We then pick an allocation maximizing the weighted Nash welfare of the agents in this subset.} 
the weighted product $\prod_{i=1}^n u_i(M_i)^{w_i}$, satisfies WWEF1 but not WEF1 \citep{ChakrabortyIgSu20}.

\citet{SegalhaleviSz19} showed that for \emph{divisible} items in the unweighted setting, MNW satisfies both resource- and population-monotonicity.
It is therefore rather surprising that the same is not true for indivisible items.

\begin{proposition}
\label{prop:MNW-negative}
In the unweighted setting, MNW satisfies neither resource-monotonicity nor population-monotonicity.
\end{proposition}

\begin{proof}
For resource-monotonicity, consider four items and two agents, both with weight~$1$, with the following utilities:
\vspace{2mm}
\begin{center}
\begin{tabular}{ c|ccc:c } 
  & Item 1 & Item 2 & Item 3 & Item 4 \\
  \hline
 Agent 1 & $3$ & $2$ & $2$ & $2$ \\
 Agent 2 & $2$ & $2$ & $1$ & $1$ \\
\end{tabular}
\end{center}
\vspace{2mm}
With only the first three items available, the unique MNW allocation gives items~$1$ and $3$ to agent~$1$ and item~$2$ to agent~$2$, resulting in a utility of $5$ for agent~$1$.
However, when we add item~$4$, MNW uniquely allocates items~$3$ and $4$ to agent~$1$ and items~$1$ and $2$ to agent~$2$.
So agent~$1$'s utility drops to $4$, violating resource-monotonicity.

For population-monotonicity, consider four items and three agents, all with weight~$1$, with the following utilities:
\vspace{2mm}
\begin{center}
\begin{tabular}{ c|cccc } 
  & Item 1 & Item 2 & Item 3 & Item 4 \\
  \hline
 Agent 1 & $2$ & $3$ & $3$ & $2$ \\
 Agent 2 & $1$ & $2$ & $1$ & $3$ \\ \hdashline
 Agent 3 & $2$ & $1$ & $1$ & $3$ \\
\end{tabular}
\end{center}
\vspace{2mm}
When only the first two agents are present, the unique MNW allocation gives items~$1$ and $3$ to agent~$1$ and items~$2$ and $4$ to agent~$2$, resulting in a utility of $5$ for agent~1.
However, when agent~$3$ joins, MNW uniquely allocates items~$2$ and $3$ to agent~$1$, item~$4$ to agent~$2$, and item~$1$ to agent~$3$.
So agent~$1$'s utility increases to $6$, violating population-monotonicity.
\end{proof}

On the other hand, we prove that MWNW fulfills weight-monotonicity, making it the only rule among the ones we consider in this paper to do so.

\begin{theorem}
\label{thm:MWNW-weightmon}
MWNW satisfies weight-monotonicity.
\end{theorem}

\begin{proof}
Fix an instance with $n,m,\textbf{w}=(w_1,\dots,w_n)$, and consider a modified instance with $w_i' > w_i$ for some $i \in N$.
First, assume that the maximum weighted Nash welfare is positive.
Suppose that MWNW chooses the allocation $\mathcal{M} = (M_1,\dots,M_n)$ and $\mathcal{M}' = (M_1',\dots,M_n')$ in the original and modified instance, respectively, and assume for contradiction that $u_i(M_i) > u_i(M_i')$.

In the original instance, the weighted Nash welfare of $\mathcal{M}$ is at least that of $\mathcal{M}'$.
The ratio of the weighted Nash welfare of $\mathcal{M}$ between the modified instance and the original instance is $u_i(M_i)^{w_i'-w_i}$, while the corresponding ratio of $\mathcal{M}'$ is $u_i(M_i')^{w_i'-w_i} < u_i(M_i)^{w_i'-w_i}$.
Since the maximum weighted Nash welfare is positive, in the modified instance, the weighted Nash welfare of $\mathcal{M}$ is strictly greater than that of $\mathcal{M}'$.
This means that MWNW should not have chosen $\mathcal{M}'$ in the modified instance, a contradiction.

Consider now the case where the maximum weighted Nash welfare is $0$.
Observe that since we break ties among the maximum sets of agents receiving positive utilities in a consistent manner, this set is the same before and after agent~$i$'s weight increases.
The conclusion follows by applying a similar argument as above to the utilities of the agents in this set.
\end{proof}

As mentioned, MWNW is known to satisfy WWEF1 but not WEF1.
To complete the picture, we show that it does not satisfy WPROP1.
This contrasts with the unweighted setting, where MNW satisfies PROP1 (since EF1 implies PROP1).

\begin{proposition}
\label{prop:MWNW-WPROP1}
MWNW does not satisfy WPROP1.
\end{proposition}

\begin{proof}
We use the same instance as for Dean's and Hill's methods in Theorem~\ref{thm:divisor-WPROP1}.
In order to make the weighted Nash welfare positive, MWNW assigns exactly one item to every agent and thus fails WPROP1.
\end{proof}

\section{Conclusion and Future Work}

In this paper, we have thoroughly investigated picking sequences derived from common apportionment methods, including the five traditional divisor methods and the quota method, in relation to fairness and monotonicity properties.
Our results indicate that picking sequences based on divisor methods provide strong guarantees in weighted fair division scenarios such as allocating ministries to political parties, with Adams' and Jefferson's methods standing out for fulfilling WEF1 and WPROP1, respectively.
Since Jefferson's method tends to favor large parties while Adams' often benefits smaller ones (see Section~\ref{subsec:relatedwork}), an interesting question is whether there are compelling fairness notions in addition to WWEF1 that the other three traditional divisor methods, which intuitively lie somewhere in the middle, satisfy.

A natural direction for future work is to construct rules that exhibit a stronger axiomatic behavior than the ones considered in this paper, or to prove that such rules do not exist.
Satisfying the three monotonicity properties simultaneously is trivial: one can always allocate all items to a fixed agent, or ignore the weights and use the round-robin algorithm with a fixed ordering.
However, this will not result in a fair allocation with respect to the weights.
Does there exist a rule fulfilling the three monotonicity properties along with, say, WWEF1?
Other notions that one could consider include strategyproofness and Pareto optimality---even in the unweighted setting, we are not aware of any rule that simultaneously fulfills EF1, Pareto optimality, and resource- or population-monotonicity.\footnote{We remark here that no rule that relies only on agents' rankings over items, including all picking sequences, can simultaneously satisfy EF1 and Pareto optimality.
To see this, consider six items and two agents who rank them in the order $1,2,\dots,6$ from best to worst.
In order to guarantee EF1, each agent must get one item from each of the sets $\{1,2\},\{3,4\},\{5,6\}$. 
However, any such allocation is not always Pareto optimal. 
Indeed, it could be that agent~$1$ has utilities $100, 99, 3, 2, 1, 0$ for the items, while agent~$2$ has utilities $100, 99, 98, 97, 96, 95$. 
Any allocation satisfying the above condition gives agent~$1$ exactly one of the first two items and agent~$2$ exactly three items, but this is Pareto dominated by the allocation that gives agent~$1$ the first two items and agent~$2$ the last four items.}
Moreover, while we showed that the five traditional divisor methods fail weight-monotonicity, it would be useful to determine whether other (perhaps non-divisor) apportionment methods satisfy this property.
Another avenue is to generalize the \emph{upper quota} notion from apportionment to the item allocation setting---this could result in a notion that disallows agents from receiving too much of the resource, to make sure that other agents would not complain.
Overall, our work leaves many intriguing combinations of properties to investigate, which we hope will lead to more interesting rules for fair resource allocation.

\section*{Acknowledgments}

This work was partially supported by the Deutsche
Forschungsgemeinschaft under grant BR 4744/2-1 and by an NUS Start-up Grant.
We would like to thank Dominik Peters for helpful discussion and the anonymous reviewers of the 30th International Joint Conference on Artificial Intelligence (IJCAI 2021), the 3rd Games, Agents, and Incentives Workshop (GAIW 2021), and Artificial Intelligence Journal for valuable feedback.

\bibliographystyle{plainnat}
\bibliography{main}

\begin{thebibliography}{40}
\providecommand{\natexlab}[1]{#1}
\providecommand{\url}[1]{\texttt{#1}}
\expandafter\ifx\csname urlstyle\endcsname\relax
  \providecommand{\doi}[1]{doi: #1}\else
  \providecommand{\doi}{doi: \begingroup \urlstyle{rm}\Url}\fi

\bibitem[Amanatidis et~al.(2021)Amanatidis, Birmpas, Filos-Ratsikas, Hollender,
  and Voudouris]{AmanatidisBiFi21}
Georgios Amanatidis, Georgios Birmpas, Aris Filos-Ratsikas, Alexandros
  Hollender, and Alexandros~A. Voudouris.
\newblock Maximum {N}ash welfare and other stories about {EFX}.
\newblock \emph{Theoretical Computer Science}, 863:\penalty0 69--85, 2021.

\bibitem[Aziz(2020)]{Aziz20}
Haris Aziz.
\newblock Developments in multi-agent fair allocation.
\newblock In \emph{Proceedings of the 34th AAAI Conference on Artificial
  Intelligence (AAAI)}, pages 13563--13568, 2020.

\bibitem[Aziz et~al.(2015)Aziz, Walsh, and Xia]{AzizWaXi15}
Haris Aziz, Toby Walsh, and Lirong Xia.
\newblock Possible and necessary allocations via sequential mechanisms.
\newblock In \emph{Proceedings of the 24th International Joint Conference on
  Artificial Intelligence (IJCAI)}, pages 468--474, 2015.

\bibitem[Aziz et~al.(2019{\natexlab{a}})Aziz, Caragiannis, Igarashi, and
  Walsh]{AzizCaIg19}
Haris Aziz, Ioannis Caragiannis, Ayumi Igarashi, and Toby Walsh.
\newblock Fair allocation of indivisible goods and chores.
\newblock In \emph{Proceedings of the 28th International Joint Conference on
  Artificial Intelligence (IJCAI)}, pages 53--59, 2019{\natexlab{a}}.

\bibitem[Aziz et~al.(2019{\natexlab{b}})Aziz, Chan, and Li]{AzizChLi19}
Haris Aziz, Hau Chan, and Bo~Li.
\newblock Weighted maxmin fair share allocation of indivisible chores.
\newblock In \emph{Proceedings of the 28th International Joint Conference on
  Artificial Intelligence (IJCAI)}, pages 46--52, 2019{\natexlab{b}}.

\bibitem[Aziz et~al.(2020)Aziz, Moulin, and Sandomirskiy]{AzizMoSa20}
Haris Aziz, Herv\'{e} Moulin, and Fedor Sandomirskiy.
\newblock A polynomial-time algorithm for computing a {P}areto optimal and
  almost proportional allocation.
\newblock \emph{Operations Research Letters}, 48\penalty0 (5):\penalty0
  573--578, 2020.

\bibitem[Babaioff et~al.(2021{\natexlab{a}})Babaioff, Ezra, and
  Feige]{BabaioffEzFe21}
Moshe Babaioff, Tomer Ezra, and Uriel Feige.
\newblock Fair-share allocations for agents with arbitrary entitlements.
\newblock In \emph{Proceedings of the 22nd ACM Conference on Economics and
  Computation (EC)}, page 127, 2021{\natexlab{a}}.

\bibitem[Babaioff et~al.(2021{\natexlab{b}})Babaioff, Nisan, and
  Talgam-Cohen]{BabaioffNiTa21}
Moshe Babaioff, Noam Nisan, and Inbal Talgam-Cohen.
\newblock Competitive equilibrium with indivisible goods and generic budgets.
\newblock \emph{Mathematics of Operations Research}, 46\penalty0 (1):\penalty0
  382--403, 2021{\natexlab{b}}.

\bibitem[Balinski and Young(1975)]{BalinskiYo75}
Michel~L. Balinski and H.~Peyton Young.
\newblock The quota method of apportionment.
\newblock \emph{American Mathematical Monthly}, 82\penalty0 (7):\penalty0
  701--730, 1975.

\bibitem[Balinski and Young(2001)]{BalinskiYo01}
Michel~L. Balinski and H.~Peyton Young.
\newblock \emph{Fair Representation: Meeting the Ideal of One Man, One Vote}.
\newblock Brookings Institution Press, 2001.

\bibitem[Beynier et~al.(2019)Beynier, Bouveret, Lema\^{i}tre, Maudet, Rey, and
  Shams]{BeynierBoLe19}
Aur\'{e}lie Beynier, Sylvain Bouveret, Michel Lema\^{i}tre, Nicolas Maudet,
  Simon Rey, and Parham Shams.
\newblock Efficiency, sequenceability and deal-optimality in fair division of
  indivisible goods.
\newblock In \emph{Proceedings of the 18th International Conference on
  Autonomous Agents and Multiagent Systems (AAMAS)}, pages 900--908, 2019.

\bibitem[Bouveret and Lang(2011)]{BouveretLa11}
Sylvain Bouveret and J\'{e}r\^{o}me Lang.
\newblock A general elicitation-free protocol for allocating indivisible goods.
\newblock In \emph{Proceedings of the 22nd International Joint Conference on
  Artificial Intelligence (IJCAI)}, pages 73--78, 2011.

\bibitem[Bouveret and Lang(2014)]{BouveretLa14}
Sylvain Bouveret and J\'{e}r\^{o}me Lang.
\newblock Manipulating picking sequences.
\newblock In \emph{Proceedings of the 21st European Conference on Artificial
  Intelligence (ECAI)}, pages 141--146, 2014.

\bibitem[Bouveret et~al.(2016)Bouveret, Chevaleyre, and Maudet]{BouveretChMa16}
Sylvain Bouveret, Yann Chevaleyre, and Nicolas Maudet.
\newblock Fair allocation of indivisible goods.
\newblock In Felix Brandt, Vincent Conitzer, Ulle Endriss, J{\'e}r{\^o}me Lang,
  and Ariel~D. Procaccia, editors, \emph{Handbook of Computational Social
  Choice}, chapter~12, pages 284--310. Cambridge University Press, 2016.

\bibitem[Brams and Kaplan(2004)]{BramsKa04}
Steven~J. Brams and Todd~R. Kaplan.
\newblock Dividing the indivisible: Procedures for allocating cabinet
  ministries to political parties in a parliamentary system.
\newblock \emph{Journal of Theoretical Politics}, 16\penalty0 (2):\penalty0
  143--173, 2004.

\bibitem[Brams and Taylor(1996)]{BramsTa96}
Steven~J. Brams and Alan~D. Taylor.
\newblock \emph{Fair Division: From Cake-Cutting to Dispute Resolution}.
\newblock Cambridge University Press, 1996.

\bibitem[Bredereck et~al.(2020)Bredereck, Faliszewski, Furdyna, Kaczmarczyk,
  and Lackner]{BredereckFaFu20}
Robert Bredereck, Piotr Faliszewski, Michal Furdyna, Andrzej Kaczmarczyk, and
  Martin Lackner.
\newblock Strategic campaign management in apportionment elections.
\newblock In \emph{Proceedings of the 29th International Joint Conference on
  Artificial Intelligence (IJCAI)}, pages 103--109, 2020.

\bibitem[Brill et~al.(2017)Brill, Laslier, and Skowron]{BrillLaSk17}
Markus Brill, Jean-Fran\c{c}ois Laslier, and Piotr Skowron.
\newblock Multiwinner approval rules as apportionment methods.
\newblock In \emph{Proceedings of the 31st AAAI Conference on Artificial
  Intelligence (AAAI)}, pages 414--420, 2017.

\bibitem[Brill et~al.(2020)Brill, G\"{o}lz, Peters, Schmidt-Kraepelin, and
  Wilker]{BrillGoPe20}
Markus Brill, Paul G\"{o}lz, Dominik Peters, Ulrike Schmidt-Kraepelin, and Kai
  Wilker.
\newblock Approval-based apportionment.
\newblock In \emph{Proceedings of the 34th AAAI Conference on Artificial
  Intelligence (AAAI)}, pages 1854--1861, 2020.

\bibitem[Budish(2011)]{Budish11}
Eric Budish.
\newblock The combinatorial assignment problem: Approximate competitive
  equilibrium from equal incomes.
\newblock \emph{Journal of Political Economy}, 119\penalty0 (6):\penalty0
  1061--1103, 2011.

\bibitem[Caragiannis et~al.(2019)Caragiannis, Kurokawa, Moulin, Procaccia,
  Shah, and Wang]{CaragiannisKuMo19}
Ioannis Caragiannis, David Kurokawa, Herv\'{e} Moulin, Ariel~D. Procaccia,
  Nisarg Shah, and Junxing Wang.
\newblock The unreasonable fairness of maximum {N}ash welfare.
\newblock \emph{ACM Transactions on Economics and Computation}, 7\penalty0
  (3):\penalty0 12:1--12:32, 2019.

\bibitem[Chakraborty et~al.(2020)Chakraborty, Igarashi, Suksompong, and
  Zick]{ChakrabortyIgSu20}
Mithun Chakraborty, Ayumi Igarashi, Warut Suksompong, and Yair Zick.
\newblock Weighted envy-freeness in indivisible item allocation.
\newblock In \emph{Proceedings of the 19th International Conference on
  Autonomous Agents and Multiagent Systems (AAMAS)}, pages 231--239, 2020.

\bibitem[Conitzer et~al.(2017)Conitzer, Freeman, and Shah]{ConitzerFrSh17}
Vincent Conitzer, Rupert Freeman, and Nisarg Shah.
\newblock Fair public decision making.
\newblock In \emph{Proceedings of the 18th ACM Conference on Economics and
  Computation (EC)}, pages 629--646, 2017.

\bibitem[Crew et~al.(2020)Crew, Narayanan, and Spirkl]{CrewNaSp20}
Logan Crew, Bhargav Narayanan, and Sophie Spirkl.
\newblock Disproportionate division.
\newblock \emph{Bulletin of the London Mathematical Society}, 52\penalty0
  (5):\penalty0 885--890, 2020.

\bibitem[Cseh and Fleiner(2020)]{CsehFl20}
\'{A}gnes Cseh and Tam\'{a}s Fleiner.
\newblock The complexity of cake cutting with unequal shares.
\newblock \emph{ACM Transactions on Algorithms}, 16\penalty0 (3):\penalty0
  29:1--29:21, 2020.

\bibitem[Farhadi et~al.(2019)Farhadi, Ghodsi, Hajiaghayi, Lahaie, Pennock,
  Seddighin, Seddighin, and Yami]{FarhadiGhHa19}
Alireza Farhadi, Mohammad Ghodsi, MohammadTaghi Hajiaghayi, Sebastien Lahaie,
  David Pennock, Masoud Seddighin, Saeed Seddighin, and Hadi Yami.
\newblock Fair allocation of indivisible goods to asymmetric agents.
\newblock \emph{Journal of Artificial Intelligence Research}, 64:\penalty0
  1--20, 2019.

\bibitem[Halpern et~al.(2020)Halpern, Procaccia, Psomas, and
  Shah]{HalpernPrPs20}
Daniel Halpern, Ariel~D. Procaccia, Alexandros Psomas, and Nisarg Shah.
\newblock Fair division with binary valuations: One rule to rule them all.
\newblock In \emph{Proceedings of the 16th Conference on Web and Internet
  Economics (WINE)}, pages 370--383, 2020.

\bibitem[Kurokawa et~al.(2018)Kurokawa, Procaccia, and Wang]{KurokawaPrWa18}
David Kurokawa, Ariel~D. Procaccia, and Junxing Wang.
\newblock Fair enough: Guaranteeing approximate maximin shares.
\newblock \emph{Journal of the ACM}, 64\penalty0 (2):\penalty0 8:1--8:27, 2018.

\bibitem[Lipton et~al.(2004)Lipton, Markakis, Mossel, and Saberi]{LiptonMaMo04}
Richard~J. Lipton, Evangelos Markakis, Elchanan Mossel, and Amin Saberi.
\newblock On approximately fair allocations of indivisible goods.
\newblock In \emph{Proceedings of the 5th ACM Conference on Economics and
  Computation (EC)}, pages 125--131, 2004.

\bibitem[Markakis(2017)]{Markakis17}
Evangelos Markakis.
\newblock Approximation algorithms and hardness results for fair division.
\newblock In Ulle Endriss, editor, \emph{Trends in Computational Social
  Choice}, chapter~12, pages 231--247. AI Access, 2017.

\bibitem[Moulin(2003)]{Moulin03}
Herv\'{e} Moulin.
\newblock \emph{Fair Division and Collective Welfare}.
\newblock MIT Press, 2003.

\bibitem[O'Leary et~al.(2005)O'Leary, Grofman, and Elklit]{OlearyGrEl05}
Brendan O'Leary, Bernard Grofman, and J\o{}rgen Elklit.
\newblock Divisor methods for sequential portfolio allocation in multi-party
  executive bodies: Evidence from {N}orthern {I}reland and {D}enmark.
\newblock \emph{American Journal of Political Science}, 49\penalty0
  (1):\penalty0 198--211, 2005.

\bibitem[Pukelsheim(2014)]{Pukelsheim14}
Friedrich Pukelsheim.
\newblock \emph{Proportional Representation: Apportionment Methods and Their
  Applications}.
\newblock Springer, 2014.

\bibitem[Segal-Halevi(2019)]{Segalhalevi19}
Erel Segal-Halevi.
\newblock Cake-cutting with different entitlements: How many cuts are needed?
\newblock \emph{Journal of Mathematical Analysis and Applications},
  480\penalty0 (1):\penalty0 123382, 2019.

\bibitem[Segal-Halevi and Sziklai(2018)]{SegalhaleviSz18}
Erel Segal-Halevi and Bal\'{a}zs~R. Sziklai.
\newblock Resource-monotonicity and population-monotonicity in connected
  cake-cutting.
\newblock \emph{Mathematical Social Sciences}, 95:\penalty0 19--30, 2018.

\bibitem[Segal-Halevi and Sziklai(2019)]{SegalhaleviSz19}
Erel Segal-Halevi and Bal\'{a}zs~R. Sziklai.
\newblock Monotonicity and competitive equilibrium in cake-cutting.
\newblock \emph{Economic Theory}, 68\penalty0 (2):\penalty0 363--401, 2019.

\bibitem[Tominaga et~al.(2016)Tominaga, Todo, and Yokoo]{TominagaToYo16}
Yuto Tominaga, Taiki Todo, and Makoto Yokoo.
\newblock Manipulations in two-agent sequential allocation with random
  sequences.
\newblock In \emph{Proceedings of the 15th International Conference on
  Autonomous Agents and Multiagent Systems (AAMAS)}, pages 141--149, 2016.

\bibitem[Walsh(2020)]{Walsh20}
Toby Walsh.
\newblock Fair division: the computer scientist’s perspective.
\newblock In \emph{Proceedings of the 29th International Joint Conference on
  Artificial Intelligence (IJCAI)}, pages 4966--4972, 2020.

\bibitem[Wintein and Heilmann(2018)]{WinteinHe18}
Stefan Wintein and Conrad Heilmann.
\newblock Dividing the indivisible: Apportionment and philosophical theories of
  fairness.
\newblock \emph{Politics, Philosophy \& Economics}, 17\penalty0 (1):\penalty0
  51--74, 2018.

\bibitem[Xiao and Ling(2020)]{XiaoLi20}
Mingyu Xiao and Jiaxing Ling.
\newblock Algorithms for manipulating sequential allocation.
\newblock In \emph{Proceedings of the 34th AAAI Conference on Artificial
  Intelligence (AAAI)}, pages 2302--2309, 2020.

\end{thebibliography}

\appendix

\section{Counterexamples for Resource-Monotonicity}
\label{app:resmon-failure}

Resource-monotonicity, as enunciated in Definition~\ref{def:mon}, appears to be an intuitive property that a ``reasonable'' allocation rule should possess. 
However, in this section we show that even in the unweighted setting, (natural versions of) two popular algorithms for fair allocation violate this axiom: the \textit{envy-cycle elimination} algorithm \citep{LiptonMaMo04}, and the \emph{adjusted winner} procedure \citep{BramsTa96}.\footnote{Note that the examples provided in the proofs of Propositions~\ref{prop:lipton_resmon} and~\ref{prop:aw_resmon} may not be counterexamples to resource-monotonicity for other versions of the respective algorithms. The aim of this section is not to produce an exhaustive list of algorithms failing resource-monotonicity, but to demonstrate the surprising non-triviality of this axiom.} 
Both of these algorithms produce EF1 allocations---the first works for any number of agents with arbitrary non-negative, monotone (not necessarily additive) utility functions, whereas the latter produces a Pareto optimal allocation for two agents with additive valuations. 
We refer to the above cited references for detailed descriptions of these algorithms.

Every iteration of the envy cycle elimination algorithm begins by allocating an arbitrary unallocated item to an arbitrary agent who is currently not envied by any other agent (the existence of such an agent is guaranteed by the envy cycle elimination step which follows item allocation)---this maintains the EF1 property as an invariant for the partial allocation at the end of each iteration. 
However, in an actual implementation, we have to use a tie-breaking convention to decide which agent gets which item in case there is a non-unique eligible agent-item pair. 
We will focus on a \textit{prima facie} reasonable tie-breaking rule: select the agent-item pair that results in the maximum gain in utility among all pairs of available items and unenvied agents, breaking further ties lexicographically with respect to agents first, and then with respect to items. We call this the \textit{maximum marginal utility tie-breaking rule}.

\begin{proposition}\label{prop:lipton_resmon}
In the unweighted setting, the envy cycle elimination algorithm with the maximum marginal utility tie-breaking rule does not satisfy resource-monotonicity.
\end{proposition}

\begin{proof}
	Consider four items and three agents, all with weight~$1$, with the following utilities: 
	\vspace{3mm}
	\begin{center}
		\begin{tabular}{ c|c:ccc } 
			& Item 1 & Item 2 & Item 3 & Item 4 \\
			\hline
			Agent 1 & $11$ & $10$ & $5$ & $1$ \\
			Agent 2 & $1$ & $6$ & $1$ & $2$ \\
			Agent 3 & $0$ & $0$ & $4$ & $1$ \\
		\end{tabular}
	\end{center}
	\vspace{3mm}
	In the instance with items $2$, $3$, and $4$ only, we start by giving item $2$ to agent $1$. 
	Now agents~$2$ and $3$ are unenvied, so agent~$3$ gets item $3$ and with it a utility of $4$. Finally, only agent~$2$ is unenvied and receives item~$4$.
	
	When we include item~$1$ in the set, the first iteration assigns item~$1$ to agent~$1$. 
	Agents~$2$ and $3$ are now eligible, so agent~$2$ takes item $2$. 
	At this point, there is no envy. 
	So agent~$1$ takes item $3$, leaving agents~$2$ and $3$ as the unenvied agents. 
	Finally, agent~$2$ gets item $4$. Since agent $3$ receives no item, her valuation drops to $0$ even though the set of items was augmented.
\end{proof}

The adjusted winner procedure in its original form applies to dividing a collection of \textit{divisible} items between two (unweighted) agents, and produces an envy-free and Pareto optimal allocation with at most one item split between the agents. 
When the items are indivisible, there are different ways to adjust the procedure to obtain an EF1 and Pareto optimal allocation---here we will focus on the variant put forward by \citet{ChakrabortyIgSu20}.

\begin{proposition}\label{prop:aw_resmon}
	In the unweighted setting, the adjusted winner procedure for indivisible items does not satisfy resource-monotonicity.
\end{proposition}
\begin{proof}
	Let $\varepsilon\in(0,1)$. Consider four items and two agents, both with weight $1$, with the following utilities:
	\vspace{3mm}
	\begin{center}
		\begin{tabular}{ c|c:ccc } 
			& Item 1 & Item 2 & Item 3 & Item 4\\
			\hline
			Agent 1 & $\varepsilon$ & $1-\varepsilon$ & $2\varepsilon$ & $1$\\
			Agent 2 & $\varepsilon$ & $\frac{3}{2}(1-\varepsilon)$ & $4\varepsilon$ & $3$ 
		\end{tabular}
	\end{center}
	\vspace{3mm}
	First, consider the instance with items~$2$, $3$, and $4$ only. The ratios of agent~$1$'s utility to that of agent~$2$ for these items are $2/3$, $1/2$, and $1/3$, respectively. 
	The procedure iterates over items~$2,3,4$ and performs the following comparisons: (1) $u_1(\{2\}) = 1 - \varepsilon < 1 = u_1(\{4\})$; (2) $u_1(\{2,3\}) = 1 + \varepsilon > 0$. Hence, the procedure terminates and assigns items~$2$ and $3$ to agent~$1$ and item~$4$ to agent~$2$.
	
	In the instance with all four items, the above ratio for item~$1$ being $1$, the procedure now iterates over the items $1,2,3,4$ and performs the following comparisons: (1) $u_1(\{1\}) = \varepsilon < 1 + 2\varepsilon= u_1(\{3,4\})$; (2) $u_1(\{1,2\}) = 1 = u_1(\{4\})$. Hence, the procedure stops and gives items $1$ and $2$ to agent~$1$ and the remaining items to agent~$2$.
	However, agent~$1$'s utility in the latter instance is $u_1(\{1,2\}) = 1$, which is lower than that in the former instance: $u_1(\{2,3\}) = 1 + \varepsilon$.
\end{proof}

\section{Counterexample for the Quota Method}
\label{app:quota-counterexample}

We present an example that the quota method is not weight-consistent.
In fact, the same example can also be used to show that the method is not weight-monotone, although Proposition~\ref{prop:quota-popmon} provides a counterexample for weight-monotonicity with a smaller number of agents.

\begin{proposition}
\label{prop:quota-weights}
The quota method fails weight-consistency and weight-monotonicity. 
\end{proposition}

\begin{proof}
Consider nine agents and seven items with the following utilities:

\vspace{3mm}
\begin{center}
\begin{tabular}{ c|ccccccc } 
  & Item 1 & Item 2 & Item 3 & Item 4 & Item 5 & Item 6 & Item 7 \\
  \hline
 Agent 1 & $3$ & $0$ & $0$ & $2$ & $0$ & $0$ & $1$ \\
 Agent 2 & $0$ & $3$ & $2$ & $0$ & $1$ & $0$ & $0$ \\
 Agent 3 & $0$ & $0$ & $2$ & $0$ & $0$ & $0$ & $1$ \\
 Agents 4--9 & $0$ & $0$ & $0$ & $0$ & $0$ & $1$ & $0$\\
\end{tabular}
\end{center}
\vspace{3mm}

In the first instance, the weights of the agents are given by $w_1=8/24,w_2=7/24,w_3=3/24$ and $w_i=1/24$ for all $i \in \{4,5,6,7,8,9\}$. Observe that the weights of the agents sum up to $1$. 
We claim that the picking sequence induced by the quota method in this example is $\pi_1 = (1,2,3,1,2,4,1)$. 
For completeness, we go through the seven rounds in the following. 
In the first round, all agents are eligible and agent~$1$ uniquely minimizes $1/w_i$. 
In the second round, all but agent~$1$ are eligible and agent~$2$ minimizes $1/w_i$ among the eligible agents. 
In the third round, agents $3$--$9$ are eligible and hence, the next pick goes to agent $3$. 
In round four, agents $1$ and $2$ become eligible again while agent $3$ is not eligible anymore, and because $2/w_1 < 1/w_4$, agent~$1$ gets assigned the next pick. 
Similarly, agent~$2$ gets another pick in round five. 
In round six, the set of eligible agents contains only agents $4$--$9$ and we assume without loss of generality that the tie is broken in favor of agent~$4$. 
Finally, agent~$1$ becomes eligible again in the last round and gets assigned the final pick. 

We modify the instance by increasing the weight of agent~$1$ to $w'_1 = 9/24$ and keeping the remaining weights unchanged. 
Observe that the weights now sum up to $25/24$. 
We claim that the picking sequence induced by the quota rule for the modified instance is $\pi_2 = (1,2,1,2,3,1,4)$. For completeness, we go through the seven rounds in the following. 
The first two rounds proceed similarly as before, i.e., agents $1$ and $2$ are selected in rounds $1$ and $2$, respectively. 
In round three, however, the increased weight of agent~$1$ leads to agent $1$ being eligible already in this round. Since, $2/w_1 < 1/w_3$, agent~$1$ is selected again. 
Similarly, agent~$2$ is eligible in round four and giving her a second pick is preferred over giving a first pick to agent~$3$. 
In round five, the set of eligible agents contains agents $3$--$9$ and hence, agent~$3$ finally gets her first pick. 
In round six, agent~$1$ is eligible and since $3/w_1 < 1/w_4$, agent~$1$ gets her third pick in this round. 
Lastly, the set of eligible agents consists of agents $4$--$9$ and agent~$4$ is selected without loss of generality by tie-breaking. 

The picking sequence $\pi_1$ induces a utility of $6$ for agent~$1$, as she can pick items $1$, $4$, and $7$. 
By contrast, the picking sequence $\pi_2$ induces a utility of only $5$ for agent~$1$. 
This is because by the time of her third pick, agent~$3$ has already chosen item~$7$, leaving agent~$1$ with items $1$ and $4$ along with one item for which she has zero utility. 
This shows that the quota method fails weight-monotonicity. Moreover, one can check that $\pi_2$ cannot be constructed from $\pi_1$ by shifting picks of agent~$1$ to the front, inserting additional picks for agent~$1$, and trimming the resulting sequence. 
Hence, the quota method also fails weight-consistency. 
\end{proof}

\end{document}